\newif\ifnotes
\notestrue

\documentclass[11pt]{article}

\usepackage[margin=1in]{geometry}
\usepackage{amsmath,amssymb,amsthm,amsfonts,mathtools}
\usepackage{microtype}
\usepackage{enumitem}
\usepackage[pagebackref]{hyperref}
\hypersetup{colorlinks=true,linkcolor=blue,citecolor=blue,urlcolor=blue}
\usepackage[dvipsnames]{xcolor}
\usepackage{colortbl}
\usepackage{xspace}

\theoremstyle{plain}
\newtheorem{theorem}{Theorem}[section]
\newtheorem{lemma}[theorem]{Lemma}
\newtheorem{proposition}[theorem]{Proposition}
\newtheorem{corollary}[theorem]{Corollary}

\theoremstyle{definition}
\newtheorem{definition}[theorem]{Definition}
\newtheorem{remark}[theorem]{Remark}

\newcommand{\F}{\mathbb{F}}

\newcommand{\Tr}{\operatorname{Tr}}

\newcommand{\abs}[1]{\left|#1\right|}
\DeclareMathOperator{\Span}{Span}

\title{Algebraic Expander Codes}
\author{Swastik Kopparty\thanks{
Department of Mathematics and Department of Computer Science, University of Toronto. Email: \url{swastik.kopparty@utoronto.ca}.
} \and Itzhak Tamo
\thanks{ Department of Electrical Engineering-Systems, Tel Aviv
University. Email: \url{tamo@tauex.tau.ac.il}. I. Tamo was supported  by the European Research Council (ERC) under Grant 852953.}}

\date{\today}

\begin{document}
\maketitle

\begin{abstract}
Expander (Tanner) codes combine sparse graphs with local constraints, enabling linear-time decoding and asymptotically good distance--rate tradeoffs. A standard constraint-counting argument yields the global-rate lower bound $R\ge 2r-1$ for a Tanner code with local rate $r$, which gives no positive-rate guarantee in the low-rate regime $r\le 1/2$. This regime is nonetheless important in applications that require algebraic local constraints (e.g., Reed--Solomon locality and the Schur-product/multiplication property).

We introduce \emph{Algebraic Expander Codes}, an explicit algebraic family of Tanner-type codes whose local constraints are Reed--Solomon and whose global rate remains bounded away from $0$ for every fixed $r\in(0,1)$ (in particular, for $r\le 1/2$), while achieving constant relative distance.

Our codes are defined by evaluating a structured subspace of polynomials on an orbit of a non-commutative subgroup of $\mathrm{AGL}(1,\mathbb{F})$ generated by translations and scalings. The resulting sparse coset geometry forms a strong spectral expander, proved via additive character-sum estimates, while the rate analysis uses a new notion of polynomial degree and a polytope-volume/dimension-counting argument.
\end{abstract}

\section{Introduction}

Expander codes, introduced by Sipser and Spielman \cite{SipserSp96}, instantiate the Tanner-code viewpoint \cite{tanner1981recursive}: a long code is specified by local constraints placed on a sparse graph.
When the underlying graph has strong expansion, this yields asymptotically good codes with efficient decoding algorithms (see, e.g., \cite{SipserSp96,zemore-expander-decoding,hlw2006expander}).

A central challenge in the design of expander codes is the inherent tension between the rate of the local code and the guaranteed rate of the global code. For an expander  code defined on a regular graph with local codes of rate $r$, the standard constraint-counting argument yields the lower bound on the global rate $R \ge 2r - 1$ \cite{tanner1981recursive}.
This bound implies a sharp threshold: to guarantee a positive global rate via this argument, the local codes must have rate $r > 1/2$. In the ``low-rate regime'' where $r \le 1/2$, the bound $2r-1$ becomes non-positive. This does not preclude positive-rate constructions, but it removes the standard black-box guarantee.

Despite this barrier, the regime of $r \le 1/2$ is of  theoretical and practical interest, particularly when the local codes possess specific algebraic structures. A primary motivation stems from the \emph{multiplication property} of Reed--Solomon (RS) codes. If $\mathcal{C}$ is an RS code of rate $r$, the component-wise product (Schur product) of two codewords belongs to an RS code of rate approximately $2r$.
This follows because the product of two polynomials of degree less than $k$ has degree less than $2k-1$, and evaluation commutes with multiplication \cite{randriambololona2015products}.
For the product code to be a non-trivial error-correcting code (i.e., rate strictly less than $1$), the initial code must satisfy $r < 1/2$.

This multiplication property is important in applications such as fault-tolerant quantum computing and High-Dimensional Expanders (HDX). For example, recent constructions of Quantum LDPC codes with transversal non-Clifford gates by Golowich and Lin \cite{golowich2024quantum} rely heavily on products of algebraic codes with rates below half. Similarly, Dinur, Liu, and Zhang \cite{dinur2023new} highlight the importance of constructing codes with the multiplication property, and construct high-dimensional expanders with this property. Thus, there is a compelling need for explicit expander codes that maintain the structure of low-rate Reed--Solomon local constraints while achieving a non-trivial global rate.

In this work, we resolve these challenges. We present an explicit construction of expander codes that achieve a positive global rate for any fixed local rate $r \in (0, 1)$, and specifically in the challenging regime of $r \le 1/2$. Crucially, our construction ensures that the local constraints are Reed--Solomon codes. This allows our codes to support the multiplication property required for the aforementioned applications while simultaneously achieving linear minimum distance via spectral expansion \cite{barg2006distance}.

\paragraph{Main theorem (informal).}
Fix any constant local rate $r\in(0,1)$. For infinitely many block lengths $n$ we give an explicit linear code $\mathcal{C}\subseteq \mathbb{F}^n$ such that:
(i) every coordinate participates in two local Reed--Solomon constraints of asymptotic rate at most $r$;
(ii) $\mathcal{C}$ has constant rate and linear (constant relative) distance; and
(iii) the underlying bipartite coset graph is a strong spectral expander (quantified by an explicit second singular value bound).

\paragraph{What is new.}
Unlike the standard “graph + local code” paradigm (where one chooses a $d$-regular expanding graph and a length-$d$ local code independently and then \emph{glues} them together via a Tanner construction), our approach is purely algebraic.
We define a single orbit-evaluation code, and the Tanner/expander graph arises as an \emph{emergent} coset geometry coming from the action of two subgroups.

We term these codes \emph{Algebraic Expander Codes} to emphasize this distinction.
A second key point is that the two subgroups are of \emph{opposing algebraic types}, translations and scalings.
Their non-commutative interaction yields a sparse underlying geometry (in contrast with the dense, grid-like geometries that arise from commuting subgroups) while still admitting strong spectral expansion.

\subsection{Connection to Locally Recoverable Codes and Graph Density}
\label{sec:connection-to-LRC}
Our construction builds upon the algebraic framework of Locally Recoverable Codes (LRCs) established by the second author and Barg \cite{tamo2014family}. In that framework, a code is defined by evaluating a special subspace of polynomials over some orbits of a subgroup $G$ of the affine group $\mathrm{AGL}(1,\mathbb{F})$. The special subspace of polynomials guarantees that the restriction of any codeword to an orbit of $G$ corresponds to a univariate polynomial of bounded degree, ensuring the local view is an RS code.

In \cite{tamo2014family}, this algebraic framework was extended to construct codes with multiple recovering sets, where each coordinate participates in multiple (typically two) local codes. To achieve this, the evaluation set is partitioned into orbits under the action of two distinct subgroups of $\mathrm{AGL}(1,\mathbb{F})$, denoted as $G$ and $H$. However, the constructions in \cite{tamo2014family} were restricted to subgroups of the \emph{same algebraic type}, either two additive subgroups representing translations in distinct directions, or two multiplicative subgroups representing two different scalings.

We note that this approach had two primary limitations. First, it did not support arbitrary local rates, requiring $r > 1/2$ to ensure a non-trivial rate, which mirrors the bounds of standard expander codes. Second, and more critically, using subgroups of the same type imposes inherent structural limitations on the code's underlying graph, which we explain next.

These codes admit a natural bipartite graph representation: code coordinates correspond to edges, while local codes (the orbits of $G$ or $H$) correspond to vertices. Specifically, an edge connects the unique $G$-orbit and $H$-orbit containing that coordinate. The degree of any vertex is simply the size of its corresponding orbit.

When $G$ and $H$ share the same algebraic type, they necessarily commute. Consequently, the resulting structure is a complete bipartite graph. While such graphs are excellent spectral expanders, their maximal degree scales linearly with the number of vertices. This contrasts with the typical objective of expander code constructions of achieving sublinear, and ideally constant, degrees. The limitation arises because commutative groups generate a group that is somewhat small relative to the subgroup sizes, constraining the code length and vertex count while maintaining high degrees.

These limitations motivate the central question of this work: what code properties arise when the $\mathrm{AGL}$ subgroups are non-commutative? We consider this avenue in this paper.

\subsection{Overview of Results}
\label{sec:overview}

Our main contribution is an explicit construction of expander codes that overcome the limitations of commutative algebraic constructions. We summarize our results.

\paragraph{Main Results.}
Fix any target local rate $r\in(0,1)$ and a sparsity parameter $m\ge 2$. 
Then for all primes $p$, we give an explicit linear code $\mathcal{C} \subseteq \F^n$,
where $n = p^{\Theta(m^2)}$ and $\F$ is a field of characteristic $p$, with the following properties:

\begin{enumerate}
    \item \textbf{Graph Structure:} $\mathcal{C}$ is a Tanner code $\mathrm{Tanner}(\Gamma, (\mathcal{C}_{v})_{v \in V})$, where: 
    \begin{itemize}
        \item $\Gamma = (V, E)$ is a bipartite \emph{spectral expander graph} 
        with $|E| = n$, with left and right degrees both $O(p^{m+1})$, and with normalized second largest singular value at most $O(1/\sqrt{p})$, 
        \item for each vertex $v \in V$, there is a code $\mathcal{C}_v \subseteq \F^{E(v)}$,
        where $E(v)$ denotes the set of edges incident on $v$.
    \end{itemize}
    and thus $\mathcal{C}$ is precisely the set of vectors $f \in \F^E$ such that
        for all $v \in V$:
        $$f|_{E(v)} \in \mathcal{C}_v.$$

    \item \textbf{Algebraic Structure:} $\mathcal{C}$ is defined by evaluating an explicit $\mathbb{F}$-linear space of polynomials on an orbit $\Omega\subseteq \mathbb{F}$ of a group $A=\langle G,H\rangle\le \mathrm{AGL}(1,\mathbb{F})$. 
    
    \item \textbf{RS Locality:} The codes $\mathcal{C}_v$ are Reed--Solomon codes of length $|E(v)|$
and dimension at most $\lceil r|E(v)|\rceil$. In particular, each coordinate participates
in two local RS constraints of local rate $r+o(1)$ and local relative distance at least $1-r-o(1)$.
Because the local constraints are RS, the code inherits the multiplication property required
for quantum and HDX applications.

    \item \textbf{Linear Global Rate:} The global rate $R$ is bounded below by an explicit positive constant depending on $r$ and $m$, even when the local rate $r\le 1/2$. This effectively bypasses the $2r-1$ rate barrier of previous Tanner code constructions (and made the multiplication property impossible to achieve).

 \item \textbf{Linear Distance:} The global relative distance satisfies $\delta \ge (1-r)^2 - o_p(1)$. This comes from the Tanner code structure and the spectral expansion of the underlying graph.
\end{enumerate}

\paragraph{Limitations.} We note two limitations of this construction. First, the alphabet size is not constant, but instead grows polynomially with the block length\footnote{In the first code instantiation in Section~\ref{sec:main-construction}, the alphabet size is in fact linear in the code length, as is typical for Reed--Solomon codes and algebraic LRCs \cite{tamo2014family}.}. Second, while the graph is sparse compared to the commutative case, the degree is polynomial in $n$ (scaling as $O\left(n^{\frac{1}{m+1}}\right)$ rather than constant. Overcoming these limitations remains an open question.

\begin{remark}
Our results are in fact more general by incorporating an additional global degree constraint on the space of evaluated polynomials (see Section~\ref{sec:code-def}). This gives rise to a family of algebraic subcodes of the previously described codes, while providing improved guarantees on the minimum distance. Importantly, these subcodes preserve the same local Reed--Solomon structure, and additionally admit a global algebraic decoding algorithm based on the imposed degree constraint.
\end{remark}
\subsection{Techniques}
To overcome the inherent obstacles in standard expander code constructions, specifically the lower bound on the local rate, we employ a purely algebraic approach.

The most significant challenge lies in proving that the global rate remains positive even when the local rate satisfies $r \le 1/2$. In this regime, the standard dimension-counting argument for expander codes fails. To address this, we introduce a new notion of polynomial degree, explained next.

\paragraph{Establishing the Rate via $u$-Base Degree.}
To lower bound the dimension of the message space, we first introduce the necessary algebraic definitions. Let $g(X)$ and $h(X)$ be polynomials invariant under the actions of $G$ and $H$, respectively. That is, $g(\phi(x)) = g(x)$ for any affine transformation $\phi \in G$ and all $x\in\mathbb{F}$, and similarly for $h(X)$ with respect to $H$.

We define the concept of the \emph{$u$-base degree} as follows: Let $u(X)$ be a polynomial of positive degree $d$. Any polynomial $f(X)$ admits a unique $u$-adic expansion $f(X)=\sum_i c_i(X)u(X)^i$, where $\deg(c_i) < d$ for all $i$. The $u$-base degree of $f$ is defined as the maximum degree among its coefficient polynomials, i.e., $\max_i \{\deg(c_i)\}$.

The task of lower bounding the code dimension is then reduced to bounding the dimension of the space of polynomials that simultaneously possess small $g$-base degree and small $h$-base degree.

A key technical contribution is proving that the $u$-base degree is sub-additive. This property allows us to translate the conditions on the message polynomials into a set of linear constraints on their $h$-base degrees. We then estimate the dimension of the message space by approximating the volume of the resulting high-dimensional polytope.
\paragraph{Spectral 
Expansion via Character Sums.}
To establish the graph expansion, we analyze the spectral gap (equivalently, the second largest singular value of the bi-adjacency matrix) by studying the two-step random walk operator. Using standard Fourier analysis techniques, we express the eigenvalues of this operator in terms of sums of non-trivial additive characters over a multiplicative subgroup \cite{alon2014additive}. By applying classical bounds on these sums (e.g., bounds on Gauss sums), we bound the second singular value of the graph by $O(1/\sqrt{p})$, where $p$ is the characteristic of the field. This confirms that the graph is a strong spectral expander.

\paragraph{Graph Sparsity via Mixed-Type Subgroups.}
As discussed in Section \ref{sec:connection-to-LRC}, constructing expander codes using the LRC framework typically yields graphs with linear degrees if one uses subgroups of the same algebraic type. We overcome this limitation by instantiating the code construction using subgroups of \emph{opposing algebraic types}. Specifically, we define the code using:
\begin{enumerate}
    \item An additive subgroup $G\leq (\mathbb{F},+)$, viewed as a translation group $\{x \mapsto x+g : g\in G\}$, and
    \item A multiplicative subgroup $H\leq (\mathbb{F}^{\times},\cdot)$, viewed as a scaling group $\{x \mapsto hx : h\in H\}$.
\end{enumerate}
Since these two subgroups of $\mathrm{AGL}(1,\mathbb{F})$ do not commute, the group $A = \langle G, H \rangle$ generated by them is significantly larger than the product of the sizes of $G$ and $H$. Consequently, the corresponding graph of the code is much sparser than in the commutative case.

\paragraph{Two explicit instantiations.}
We instantiate the general framework in two complementary ways.
The first instantiation (Section~\ref{sec:main-construction}) takes the simplest multiplicative group $H=\mathbb{F}_{p^m}^\times$ (the entire multiplicative group of a field), yielding a clean spectral analysis, that follows from the orthogonality of additive characters.
This choice also makes the two local lengths (in other words, the degree of the vertices on each side) as balanced as possible (they differ by one). However, the construction requires working over a field large enough to realize the chosen additive subgroup $G$. See Section \ref{sec:first-code} for the definition of the group $G$ and more details.

 The second instantiation (Section~\ref{sec:second-instantiation}) keeps the left local code algebraically simple by taking $G=\mathbb{F}_{p^m}$ and chooses $H$ as a subgroup of $\mathbb{F}_{p^{m+1}}^\times$ of constant index. This introduces a tunable parameter $0<\gamma<1$ that governs the right degree, $|H|=\gamma(p^{m+1}-1)$. The strong expansion  follows  via standard Gauss sum bounds. Unlike the first instantiation, this construction yields a graph with asymmetric degrees. This demonstrates the versatility of the general construction, showing that it can accommodate different group structures and rely on alternative character sum bounds (e.g., Gauss sums vs. orthogonality).

Table \ref{tab:two-instantiations} provides a comparison between the two code instantiations.
\begin{table}[t]
\centering
\scriptsize
\begin{tabular}{lcc}
\hline
 & \textbf{Instantiation I (Sec.~\ref{sec:main-construction})} & \textbf{Instantiation II (Sec.~\ref{sec:second-instantiation})} \\
\hline
Additive subgroup $G$ & roots of $X^{p^m}+X^p+X$ & $\mathbb{F}_{p^m}$ \\
Multiplicative subgroup $H$ & $\mathbb{F}_{p^m}^\times$ & subgroup of $\mathbb{F}_{p^{m+1}}^\times$  \\
Left/right degrees & $|G|=p^m$, $|H|=p^m-1$ & $|G|=p^m$, $|H|=\gamma(p^{m+1}-1),\ 0<\gamma<1$ \\
Character sum bound & Orthogonality of additive characters & Gauss-sum  \\
\hline
\end{tabular}
\caption{Two instantiations of the algebraic expander-code framework.}
\label{tab:two-instantiations}
\end{table}

\subsection{Related Work}
\paragraph{Expander/Tanner codes and expander-based coding.}
Expander codes originate in the work of Sipser and Spielman~\cite{SipserSp96}, building on the Tanner-code viewpoint~\cite{tanner1981recursive}.
A robust theory relates global distance and decoding to expansion (see, e.g.,~\cite{hlw2006expander,zemore-expander-decoding,barg2006distance}).
In the broader TCS literature, sparse-graph codes also appear prominently in the context of linear-time erasure resilience and lossless expanders (e.g.,~\cite{AlonEdmondsLuby95,AlonLuby96,CapalboReingoldVadhanWigderson02}).

\paragraph{Algebraic locality via group actions.}
On the algebraic side, our construction builds most directly on the orbit-evaluation framework for Locally Recoverable Codes introduced by the second author  and Barg~\cite{tamo2014family}, which enforces RS-type local views by restricting message polynomials relative to invariants of an affine-group action.
A key distinction is geometric: the multiple-recovering-set extensions in~\cite{tamo2014family} use pairs of subgroups of the same algebraic type, and hence commuting actions, which naturally lead to dense (grid-like) bipartite geometries.
Our work instead uses two subgroups of \emph{opposing} algebraic types (translations and scalings), whose non-commutative interaction yields a sparse coset geometry while still admitting strong spectral expansion.
Related symmetry- and group-action-based perspectives on sparse codes also appear in, e.g.,~\cite{kaufman-lubotzky}.

\paragraph{Multiplicative structure and applications.}
Our emphasis on Reed--Solomon locality is motivated by the Schur-product (multiplication) behavior of RS and related algebraic codes~\cite{randriambololona2015products}.
This property is central in recent applications, including quantum-code constructions via products of algebraic codes~\cite{golowich2024quantum} and recent HDX/code constructions that explicitly require multiplicative structure~\cite{dinur2023new}.
More broadly, expander-based outer structures have been used to leverage RS algebraic structure for algorithmic coding tasks (see, e.g.,~\cite{GuruswamiRudra05}).

\bigskip
\paragraph{Organization of the paper.}
Section~\ref{sec:general} introduces the general algebraic framework, defines the expander code, and establishes the associated bipartite graph structure.
In Section~\ref{sec:graph-analysis}, we analyze the spectral expansion of this graph by reducing the bound on the second eigenvalue to bounds on additive character sums.
The two explicit code instantiations are presented in Section~\ref{sec:main-construction} and Section~\ref{sec:second-instantiation}.
Finally, Section~\ref{sec:conclusion} offers concluding remarks and discusses open problems.


\section{General Code Construction}\label{sec:general}

\subsection{Algebraic Setup and Evaluation Domain}
\label{Algebraic_Setup}
Let $\mathbb{F}$ be a finite field of characteristic $p$. 
Throughout, we use $X$ to denote a formal indeterminate in polynomial rings (e.g., $\mathbb{F}[X]$), and we use $x\in\mathbb{F}$ to denote a field element (an evaluation point) on which affine transformations act.

Let
\[
\mathrm{AGL}(1,\mathbb{F})=\{\,x\mapsto ax+b : a\in \mathbb{F}^{\times},\ b\in \mathbb{F}\,\}
\]
denote the group of affine transformations acting on $\mathbb{F}$. Let $G,H\le \mathrm{AGL}(1,\mathbb{F})$ be subgroups, and let
\[
A=\langle G,H\rangle
\]
be the subgroup they generate.
In our main instantiations we will take $G$ to be a (pure) translation subgroup and $H$ to be a (pure) scaling subgroup. See definitions ahead.

We assume without loss of generality that $|\mathbb{F}|\ge |A|$ (replacing $\mathbb{F}$ by a sufficiently large extension field if necessary). Under this assumption, there exists an element $\alpha\in \mathbb{F}$ whose stabilizer in $A$ is trivial, i.e.,
$
\phi(\alpha)\neq \alpha \text{ for all } \phi\in A\setminus\{\mathrm{id}\}.
$
Indeed, any non-identity affine map $\phi(x)=ax+b$ has at most one fixed point in $\mathbb{F}$. Hence, the set of elements of $\mathbb{F}$ fixed by some nontrivial element of $A$ has cardinality at most $|A|-1$, and since $|\mathbb{F}|\ge |A|$, we can choose $\alpha$ outside this set.

Let $\Omega$ denote the orbit of $\alpha$ under the action of $A$:
\[
\Omega=\{\phi(\alpha): \phi\in A\}\subseteq \mathbb{F}.
\]
By construction, the stabilizer of $\alpha$ is only $\mathrm{id}$, and therefore the action of $A$ on $\Omega$ is free. Since $\Omega$ is an $A$-orbit, the action is also transitive, and hence it is regular. 

We will construct an evaluation code over $\mathbb{F}$ by evaluating at the points of $\Omega$. Consequently, the block length of the code is
$
n := |\Omega| = |A|.
$

We will next need the notion of invariant polynomials.
 
\paragraph{Invariant polynomials.}
Let $G \le \mathrm{AGL}(1,\mathbb{F})$, and let $f(X)\in \mathbb{F}[X]$. We say that $f$ is \emph{$G$-invariant} if
\[
f(g(x)) = f(x)\qquad\text{for all } g\in G\text{ and all } x\in\mathbb{F}.
\]
In particular, any $G$-invariant polynomial is constant on $G$-orbits in $\mathbb{F}$. Indeed, for every $a\in \mathbb{F}$ and $g\in G$,
\[
f(g(a)) = f(a).
\]

Clearly, every constant polynomial is $G$-invariant, but there are also nontrivial examples. For instance, it is easy to verify  that
\begin{equation}
\label{eq:invariant-poly}
\prod_{g\in G} g(X),
\end{equation}
is $G$-invariant.

We next consider two special cases: when $G$ is a pure translation group and when $G$ is a pure scaling group.

\smallskip
\noindent\textbf{Pure translations.}
Suppose
\[
G=\{\,x\mapsto x+u : u\in U\,\}
\]
for some additive subgroup $U\le (\mathbb{F},+)$. Then \eqref{eq:invariant-poly} becomes
\begin{equation}
\label{eq:invariant-poly2}
\prod_{g\in G} g(X) \;=\; \prod_{u\in U} (X+u) \;=\; \prod_{u\in U}(X-u),
\end{equation}
where the last equality holds since $U$ is closed under negation. The right-hand side of \eqref{eq:invariant-poly2} is precisely the \emph{annihilator polynomial} (also known as the \emph{subspace polynomial}) of the additive subgroup $U$.

\smallskip
\noindent\textbf{Pure scalings.}
Suppose
\[
G=\{\,x\mapsto ux : u\in U\,\}
\]
for some multiplicative subgroup $U\le (\mathbb{F}^\times,\cdot)$. Then \eqref{eq:invariant-poly} becomes
\begin{equation}
\label{eq:invariant-poly3}
\prod_{g\in G} g(X)
\;=\;
\prod_{u\in U} (uX)
\;=\;
X^{|U|}\prod_{u\in U} u.
\end{equation}
Since $\prod_{u\in U}u\in \mathbb{F}^\times$ is a constant, \eqref{eq:invariant-poly3} implies that the monomial $X^{|U|}$ is constant on the $G$-orbits of $\mathbb{F}$ under the scaling action.

\smallskip
\noindent
In the sequel, we take $G$ to be a pure translation group and define its associated $G$-invariant polynomial by
\begin{equation}
\label{eq:additive-subgroup-poly}    
g(X) \;:=\; \prod_{u\in U}(X-u),
\end{equation}

i.e., the right-hand side of \eqref{eq:invariant-poly2}. We will freely identify the translation subgroup $G\le\mathrm{AGL}(1,\mathbb{F})$ with its set of translation parameters $U\le(\mathbb{F},+)$.
Likewise, we take $H$ to be a pure scaling group and define its associated $H$-invariant polynomial by
\begin{equation}
    \label{eq:mult-subgroup-poly}
h(X) \;:=\; X^{|H|}.
\end{equation}
We will repeatedly use that $g(X)$ and $h(X)$ are constant on the corresponding $G$- and $H$-orbits in $\mathbb{F}$.

\subsection{\texorpdfstring{Base $u$ degree of a polynomial}{Base u degree of a polynomial}}

We need the following notion of the degree of a polynomial.

\begin{definition}[Base-$u$ Degree]
    Let $u(X) \in \mathbb{F}[X]$ be a \emph{non-constant} polynomial. Any polynomial $f(X) \in \mathbb{F}[X]$ can be uniquely expanded in base $u(X)$ as:
    \[
        f(X) = \sum_{i} c_i(X) u(X)^i,
    \]
    where $\deg(c_i) < \deg(u)$ for all $i$. We define the \emph{$u$-base degree} of $f$, denoted $\deg_u(f)$, as the maximum degree of its coefficients:
    \[
        \deg_u(f) := \max_i \{ \deg(c_i) \},
    \]
    with the convention $\deg(0):=-\infty$.
\end{definition}
Then, clearly, $\deg_u(f)\in \{-\infty,0,\ldots,\deg(u)-1\}$ and $\deg_u(f+g)\leq \max\{\deg_u(f),\deg_u(g)\}$. 
The following lemma shows that the $u$-base degree is sub-additive. 
\begin{lemma}[Subadditivity of the $u$-base degree]
\label{lem:base-u-subadditive}
Let  $u(X)\in\F[X]$ be a nonconstant polynomial.
Then for all $f,g\in\F[X]$,
\[
    \deg_u(fg)\ \le\ \deg_u(f)+\deg_u(g).
\]
\end{lemma}

\begin{proof}
 Let $d:=\deg(u)$.
If $\deg_u(f) + \deg_u(g) \geq d$, then the result is trivially true (since $\deg_u$ of any
polynomial is at most $d - 1$).

Thus assume that $\deg_u(f) + \deg_u(g) < d$, and write
\[
    f=\sum_{i} c_i u^i,\qquad g=\sum_{j} d_j u^j,
\]
with $\deg(c_i),\deg(d_j)<d$. Then
\begin{align}
\label{eqfg}
    fg=\sum_{k\ge 0} e_k\,u^k,
\end{align} where 
$
    e_k := \sum_{i+j=k} c_i d_j \in \F[X]$.
Now, for each $k$, we have
\[
    \deg(e_k)\ \le\ \max_{i+j=k}\big(\deg(c_i)+\deg(d_j)\big)\ \le\ \deg_u(f)+\deg_u(g) < d,
\]
and thus \eqref{eqfg} is the base $u$ expansion of $fg$, and the result follows.
\end{proof}

\subsection{Message Space and Code Definition}
\label{sec:code-def}
In this section we give the definition of the message space and the code.
This is specified by a local rate parameter $r \in (0,1)$, 
subgroups $G$ and $H$ of $\mathrm{AGL}(1, \F)$, and a global degree parameter 
$D \leq n$.
 \begin{remark}
The primary case of interest in this paper is $D = n$, in which the code coincides with a Tanner code defined on the coset graph. The more general regime $D \leq n$ yields a natural family of algebraic subcodes of this Tanner code. 

For a first reading, we therefore recommend focusing on the case $D = n$, as it simplifies the construction while capturing all the main results described in Section~\ref{sec:overview}.
\end{remark}

For the groups $G$ and $H$, let $g(X)$ and $h(X)$ be the associated $G$- and $H$-invariant polynomials, respectively.

We define the message space of polynomials $\mathcal{W}_{r,D}$ by
\begin{equation}
\label{eq:subspace-def}
\mathcal{W}_{r,D}
:=
\left\{
f(X)\in \mathbb{F}[X]\;\Big|\;
\deg(f)<D,\;\;
\deg_{g}(f)< r|G|,\;\;
\deg_{h}(f)< r|H|
\right\}
\footnote{It is straightforward to verify that $\mathcal{W}_{r,D}$ is an $\mathbb{F}$-linear subspace of $\mathbb{F}[X]$.}.
\end{equation}

The evaluation code $\mathcal{C}$ is defined as the image of the message space under evaluation over $\Omega$:
\begin{equation}
    \label{eq:main-code-construction}    
    \mathcal{C} = \mathcal{C}(G,H; r,D) := \big\{ (f(\beta))_{\beta \in \Omega} : f \in \mathcal{W}_{r,D} \big\}.
\end{equation}

We collect a few fundamental properties of $\mathcal{C}$. The code length is $n$. Since every message polynomial has degree strictly less than $D \le n$, the evaluation map is injective. Thus, $\dim_\mathbb{F}(\mathcal{C}) = \dim_\mathbb{F}(\mathcal{W}_{r,D})$. Furthermore, since $\mathcal{C}$ is a linear subcode of the Reed--Solomon code of length $n$ and degree $D$, its minimum distance is at least $n - D + 1$.

\begin{proposition}
    The code $\mathcal{C}$ is a linear block code with parameters $[\,n,\ \dim_\mathbb{F}(\mathcal{W}_{r,D}),\ \ge n-D+1\,]_\mathbb{F}$.
\end{proposition}
Our construction naturally admits two complementary decoding viewpoints. In the principal case $D = n$, the code is a Tanner/expander code and therefore inherits the standard iterative (combinatorial) decoding paradigm of expander codes. In the more general regime $D \leq n$, the code remains a subcode of the $D = n$ construction and thus still supports iterative decoding; however, it additionally admits an algebraic decoding algorithm, as it can be viewed as a subcode of a Reed--Solomon code. 

Thus, the parameter $D$ interpolates between two structural perspectives on the same construction: a combinatorial one, based on expander decoding, and an algebraic one, based on global polynomial constraints.

\begin{remark}
\label{rem:counting-bound}
    We emphasize that we are primarily interested in the regime where the local rate satisfies $r \le 1/2$. In this case, establishing a lower bound on the global rate is non-trivial. In contrast, if the local rate satisfies $r > 1/2$, a standard dimension counting argument  suffices to show that the global rate of the code is at least $(2r-1)D/n -o(1)$.

    Indeed, it suffices to show that the dimension of the subspace $\mathcal{W}_{r,D}$ is at least $(2r-1)D-2$. 
    It is easy to verify that $\mathcal{W}_{r,D}=U\cap V$, where
    \[
    \begin{aligned}
    U&=\Span\big(X^ig(X)^j:\deg(X^ig(X)^j)<D,\ i<r\abs{G} \big), \\    V&=\Span\big(X^ih(X)^j:\deg(X^ih(X)^j)<D,\ i<r\abs{H} \big),
    \end{aligned}
    \]
    and 
    $\dim(U),\dim(V)\geq \lfloor Dr\rfloor\geq Dr-1$.  Therefore,
    $$D\geq \dim(U+V)=\dim(U)+\dim(V)-\dim(U\cap V)\geq 2rD-2-\dim(\mathcal{W}_{r,D}).$$
    By rearranging the claim follows. 
\end{remark}

\subsection{Local Codes}
\label{sec:local-codes}
Since $\Omega$ is $A$-invariant, it is partitioned into orbits under the actions of both $G$ and $H$. We will show that the restriction of any codeword to each such orbit lies in a Reed--Solomon code of small rate. Equivalently, every code coordinate participates in two local Reed--Solomon constraints (one induced by a $G$-orbit and one induced by an $H$-orbit), each of small rate.

Let $\beta \in \Omega$. The orbit of $\beta$ under $G$, denoted $\beta^G$, has size $|G|$. As noted in Section \ref{Algebraic_Setup}, $g(x)$ takes a constant value, say $\lambda=g(\beta)$, on $\beta^G$. 
Let $f \in \mathcal{W}_{r,D}$ be a message polynomial. Expanding $f$ in $g$-base, we have $f(X) = \sum_i c_i(X) g(X)^i$. 
Restricting $f$ to the domain $\beta^G$, this becomes:
\[
    f|_{\beta^G}(x) = \sum_i c_i(x) \lambda^i.
\]
The condition $\deg_g(f) < r|G|$ implies that $\deg(c_i) < r|G|$ for all $i$. Therefore, the restriction of the codeword to $\beta^G$ lies in a Reed--Solomon code
of length $|G|$ and dimension at most $\lceil r|G|\rceil$. Equivalently, its local rate
is at most $\lceil r|G|\rceil/|G|=r+o(1)$, and its local relative distance is at least $1-r-o(1)$.

A symmetric argument applies to the subgroup $H$. Since $\deg_h(f) < r|H|$ and $h(x)$ is constant on $\beta^H$, the restriction of the codeword to any orbit $\beta^H$ belongs to a Reed--Solomon code of length $|H|$ and rate at most $r+o(1)$.

\subsection{Bipartite Expander-Code Viewpoint}
\label{subsec:expander-view}

The locality structure described above induces an underlying graph topology. In this section, we formalize this viewpoint by identifying the code coordinates (the elements of $\Omega$) with the edges of a bipartite graph, and the local constraints with its vertices.

More specifically, each local constraint corresponds to the restriction of a codeword to an orbit of $\Omega$ under the action of $G$ or $H$. We associate a vertex with each such orbit, and we associate an edge with each code coordinate in $\Omega$. An edge (an element of  $\Omega$) is incident to the two vertices corresponding to the unique $G$-orbit and $H$-orbit that contain it. This naturally yields a bipartite graph whose left and right vertex sets are the $G$-orbits and $H$-orbits in $\Omega$, respectively.

Next, we view this graph as a coset graph of the subgroups $G$ and $H$ within $A$. Since the action of $A$ on $\Omega$ is regular, we may identify the code coordinates (and hence the graph edges) with the elements of $A$ via the bijection $\iota:A\to\Omega$ given by
\begin{equation}
\label{eq:iota_map}
\iota(\phi) = \phi(\alpha).
\end{equation}

To determine the vertex sets, we identify which subsets of coordinates correspond to single local constraints. Fix an affine map $\phi\in A$ and let $P=\phi(\alpha)$. The $G$-orbit of $P$ is
\[
\mathrm{Orb}_G(P) \;=\; \{\, g(P) : g\in G \,\}
\;=\;
\{\, g(\phi(\alpha)) : g\in G \,\}.
\]
Since the group operation in $A$ is composition, $g(\phi(\alpha))=(g\circ\phi)(\alpha)$. Hence, under the identification $\iota$, the $G$-orbit corresponds to the set
$
G\phi \;:=\; \{\, g\phi : g\in G \,\}$, 
namely the right coset of $G$ in $A$ containing $\phi$. Consequently, the vertices corresponding to $G$-orbits are naturally indexed by the set of right cosets
\[
A/G \;=\; \{\,G\phi : \phi\in A\,\}.
\]
An identical argument shows that the vertices corresponding to $H$-orbits are indexed by the set of right cosets $A/H$.

This leads to the following graph definition.

\paragraph{Graph definition.}
We define the bipartite graph $\mathsf{B}(G,H)$ using the coset geometry of $A$ as follows:
\begin{itemize}
  \item \textbf{Left vertices ($V_L$):} the set of right cosets of $G$ in $A$, i.e.,
  $
  V_L := \{\,G\phi : \phi\in A\,\}.$
  
  \item \textbf{Right vertices ($V_R$):} the set of right cosets of $H$ in $A$, i.e.,
  $
  V_R := \{\,H\phi : \phi\in A\,\}.
  $
  \item \textbf{Edges:} the edge set is identified with the group $A$. For each $\phi\in A$, we add an edge connecting the left vertex $G\phi$ to the right vertex $H\phi$.
\end{itemize}

The graph is biregular. Indeed, each left vertex has degree $|G|$ and each right vertex has degree $|H|$. Equivalently, $\mathsf{B}(G,H)$ is the standard coset incidence graph between $A/G$ and $A/H$.

When $D = n$ the code $\mathcal{C}$ is the Tanner code on $\mathsf{B}(G,H)$ with the following local constraints:
\begin{itemize}
    \item For every left vertex $u=G\phi\in V_L$, the incident edges are indexed by the coset $G\phi$, and the restriction of a codeword to these edges must belong to a Reed--Solomon code of length $|G|$ and dimension at most $\lceil r|G|\rceil$.
  \item For every right vertex $v=H\phi\in V_R$, the incident edges are indexed by the coset $H\phi$, and the restriction of a codeword to these edges must belong to a Reed--Solomon code of length $|H|$ and dimension at most $\lceil r|H|\rceil$.
\end{itemize}
For general $D \leq n$, the code $\mathcal{C}$ is a subcode of this Tanner code.

This representation enables the use of standard expander-code tools: bounds on the minimum distance and efficient decoding can be derived from spectral expansion of $\mathsf{B}(G,H)$, quantified by the second-largest singular value of its (normalized) adjacency operator.

\begin{remark}
    
    Note that the graph is not a multigraph (i.e., it is a simple graph) if and only if $G$ and $H$ intersect only trivially. Indeed, consider the two cosets $Ha,Ga$ that contain the element $a\in A.$ Then the number of edges between the corresponding vertices $Ha$ and $Ga$ is equal to $|Ha\cap Ga|=|H\cap G|$ and the claim follows.  
\end{remark}

\section{Spectral Expansion of the Coset Graph}
\label{sec:graph-analysis}
To prove that the constructed code is indeed an expander code, it suffices to show that the associated coset graph is expanding.

In this section we analyze the spectral properties of the coset (bipartite) graph $\mathsf{B}(G,H)$ arising from the construction. In particular, we derive an upper bound on the second largest singular value of its normalized bi-adjacency operator (equivalently, the second largest eigenvalue in absolute value of the normalized adjacency matrix of the bipartite graph). This spectral bound quantifies the expansion of $\mathsf{B}(G,H)$, and will allow us to lower bound the relative minimum distance of the global code in terms of the expansion of the graph and the minimum distances of the local codes.

By abuse of notation, we let $\mathsf{B}$ also denote the rectangular bi-adjacency matrix indexed by the set of left vertices $A/G$ and right vertices $A/H$, where the entry $(Ga, Hb)$ equals the number of edges connecting the cosets, i.e., $|Ga \cap Hb|$.

We consider the \emph{normalized} adjacency operator $T: \ell_2(A/H) \to \ell_2(A/G)$ defined by:
\[
    T := \frac{1}{\sqrt{|G||H|}}\mathsf{B}.
\]
$\sigma_2(T)$ is the second largest singular value of the normalized bi-adjacency operator, and equals the second largest eigenvalue in absolute value of the normalized adjacency matrix
$\begin{pmatrix}
0 & T \\
T^* & 0 
\end{pmatrix}
$.
It is well known that the largest singular value is $\sigma_1(T) = 1$, corresponding to the constant eigenfunction $\mathbf{1}$. Thus, $\sigma_2(T)$ satisfies:
\[
    \sigma_2(T) = \sup_{\substack{f \perp \mathbf{1} \\ \|f\|=1}} \|Tf\|.
\]
To bound $\sigma_2(T)$, we analyze the eigenvalues of the self-adjoint operator $T^* T$, where it holds  that $\sigma_2(T)^2 = \lambda_2(T^* T)$.

\subsection{Spectral Analysis via Random Walks}

The operator $T^* T$ acts on the space $\ell_2(A/H)$. The entry $(T^*T)_{Hj, Hj'}$ corresponds to the probability weight of a path of length 2 starting at $Hj$ and ending at $Hj'$ in the bipartite graph. Specifically,
\[
    (T^*T)_{Hj,Hj'} = \frac{1}{|G||H|} \sum_{Gi \in A/G} |Gi \cap Hj| \cdot |Gi \cap Hj'|.
\]
For a function $f: A/H \to \mathbb{C}$, the action of the operator is given by:
\begin{equation}
    \label{eq:operator-action}
    (T^*T f)(Hj) = \sum_{Hj'\in A/H} (T^*T)_{Hj,Hj'} f(Hj')
    = \sum_{Hj'\in A/H}
    \frac{1}{|G||H|} \sum_{Gi \in A/G} |Gi \cap Hj| \cdot |Gi \cap Hj'|
    f(Hj').
\end{equation}
This summation can be interpreted as a random walk on the graph. Starting from a right coset $Hj$, the walk proceeds as follows:
\begin{enumerate}
    \item Choose a random edge incident to $Hj$. This corresponds to picking a random $h \in H$ and moving to the element $x = hj$.
    \item Traverse to the left vertex containing $x$, which is $G(hj)$.
    \item Choose a random edge incident to $G(hj)$. This corresponds to picking a random $g \in G$ and moving to the element $y = g(hj)$.
    \item Traverse to the right vertex containing $y$, which is $H(ghj)$.
\end{enumerate}
Thus, the operator \eqref{eq:operator-action} acts as an averaging operator over the $H$-cosets:
\begin{equation}
    \label{eq:avg-operator}
    (T^* T f)(Hj) = \mathbb{E}_{\substack{g \in G \\ h \in H}} \big[ f(Hghj) \big].
\end{equation}

\subsection{Analysis for Translation and Scaling Subgroups}
In this section, we analyze the operator \eqref{eq:avg-operator} for specific groups $H,G$, as follows.

\begin{itemize}
    \item Let $H \subseteq \mathbb{F}^{\times}$ be a multiplicative subgroup. We identify $H$ with the group of scaling maps $\{x \mapsto hx \mid h \in H\}$.
    \item Let $G \subseteq (\mathbb{F}, +)$ be an additive subgroup. We identify $G$ with the group of translations $\{x \mapsto x+g \mid g \in G\}$.
\end{itemize}
Before proceeding to characterize the group $A=\langle G,H\rangle$ generated by $G$ and $H$, we note that for our choice of $G$ and $H$ the resulting graph is simple. Indeed, $G$ and $H$ intersect trivially, $
G\cap H=\{\mathrm{id}\}$,
since the identity map is the only affine transformation that is simultaneously a translation and a scaling.

\begin{theorem}
\label{thm:structural-theorem}
Let $S$ be the additive closure of $G$ under the action of $H$, defined as,
\[
    S = \sum_{h \in H} hG = \left\{ \sum_{k} h_k g_k \;\middle|\; h_k \in H, g_k \in G \right\} \subseteq \mathbb{F}.
\]
Then the group $A$ is isomorphic to the semidirect product $S \rtimes H$. Specifically:
\begin{enumerate}
    \item The elements of $A$ are exactly the affine maps $\{ x \mapsto hx + s \mid h \in H, s \in S \}$.
    \item The group operation corresponds to $(s_1, h_1) \cdot (s_2, h_2) = (s_1 + h_1 s_2, h_1 h_2)$.
\end{enumerate}
\end{theorem}

\begin{proof}
See Appendix \ref{proofs-section-3}.
\end{proof}

The following lemma gives another characterization of the set $S$ that will be useful later. 
\begin{lemma}
\label{lem:S-span}
For $H$ and $G$ as above, the set $S
= \sum_{h\in H} hG
$ satisfies 
$S = \Span_{\mathbb F_p(H)}(G)$
\end{lemma}

\begin{proof}
    See Appendix \ref{proofs-section-3}.
\end{proof}

Following Theorem \ref{thm:structural-theorem}, we identify the right cosets of $H$ in $A$ with the elements of the subspace $S$. Specifically, the coset containing $(s,1)$ is $H(s,1) = \{(0,h)(s,1) : h \in H\} = \{(hs, h) : h \in H\}$.
This identification allows us to view the operator $T^* T$ as acting on the space $\mathbb{C}^S$. Given $f \in \ell_2(A/H)$, we define $f' \in \mathbb{C}^S$ by $f'(s) := f(H(s,1))$.

We trace the random walk in \eqref{eq:avg-operator} starting from the coset $H(s,1)$:
\[
    H(s,1) \xrightarrow{x \in H(s,1)} Gx \xrightarrow{y \in Gx} Hy.
\]
A generic element in $H(s,1)$ is $(hs, h)$ for a random $h \in H$. This lies in the right coset $G(hs, h) = \{(g+hs, h) : g \in G\}$. A generic element in this right coset is $(g+hs, h)$.
We must identify the right coset $H(s', 1)$ containing $(g+hs, h)$.
Since $(g+hs, h) = (0,h)(h^{-1}(g+hs), 1) = (0,h)(h^{-1}g + s, 1)$, the new state is $s' = s + h^{-1}g$.
Thus, the operator acts on $\mathbb{C}^S$ as:
\begin{equation}
    \label{eq:conv-operator}
    (T^* T f')(s) = \mathbb{E}_{\substack{g \in G \\ h \in H}} \left[ f'(s + h^{-1}g) \right].
\end{equation}

Next, for an arbitrary subset $M\subseteq \mathbb{F}$, define its trace-dual subspace by
\[
M^\perp \;:=\; \{\,a\in \mathbb{F} : \Tr(am)=0 \text{ for all } m\in M\,\},
\]
where $\Tr=\Tr_{\mathbb{F}/\mathbb{F}_p}$ denotes the field trace from $\mathbb{F}$ to its prime subfield $\mathbb{F}_p$.
For an additive subgroup (subspace) $S\le (\mathbb{F},+)$, the additive characters of $S$ are indexed by cosets in $\mathbb{F}/S^\perp$. For $a\in \mathbb{F}$ define
\[
\chi_a(s)\;:=\;\exp\!\left(\frac{2\pi i}{p}\Tr(as)\right)\qquad\text{for all } s\in S.
\]
Note that $\chi_a$ depends only on the coset $a+S^\perp$, since if $a'\equiv a \pmod{S^\perp}$ then $\Tr((a-a')s)=0$ for all $s\in S$, and hence $\chi_{a'}(s)=\chi_a(s)$.

These characters form an eigenbasis for the operator $T^*T$. Indeed, fixing a character $\chi_a$ and substituting into \eqref{eq:conv-operator}, we obtain
\[
(T^*T\,\chi_a)(s)
=
\mathbb{E}_{g,h}\big[\chi_a(s+h^{-1}g)\big]
=
\chi_a(s)\cdot \mathbb{E}_{g,h}\big[\chi_a(h^{-1}g)\big],
\]
where we used the homomorphism property $\chi_a(x+y)=\chi_a(x)\chi_a(y)$.
This implies that the eigenvalue of $\chi_a$ is $\lambda_a = \mathbb{E}_{h \in H} \mathbb{E}_{g \in G} [\chi_a(h^{-1}g)]$.
 The inner expectation is $1$ if $h^{-1}a \in G^\perp$ and $0$ otherwise.
 Thus the eigenvalue $\lambda_a$ equals
\[
    \lambda_a = \Pr_{h \in H} (h^{-1}a \in  G^\perp).
\]
The second largest singular value $\sigma_2(T)$ corresponds to the square root of the second largest eigenvalue. We conclude that,
\begin{equation}
\label{eq:sigma_2_characterization}
    \sigma_2(T) = \sqrt{\max_{a \in \mathbb{F}\setminus S^\perp} \Pr_{h \in H} (h^{-1}a \in  G^\perp)}.
\end{equation}

\subsection{Explicit Bounds via Character Sums}
In this section we bound $\sigma_2(T)$ via character-sum estimates. More specifically, we derive an explicit upper bound on $\sigma_2(T)$ in terms of additive character sums over the multiplicative subgroup $H$.

\begin{proposition}
\label{prop:general-bound}
    Let 
    \[
        M := \max_{a\in \mathbb{F}\setminus H^\perp } \left| \sum_{h \in H} \chi_a(h) \right|,
    \]
    Then, the second largest singular value of the graph satisfies
    \[
        \sigma_2(T) \le \sqrt{ \frac{1}{p} +  \frac{M}{|H|} }.
    \]
\end{proposition}

\begin{proof}
   By \eqref{eq:sigma_2_characterization} we seek to bound $\lambda_a = \frac{1}{|H|} |\{ h \in H : h^{-1}a \in G^\perp \}|$ for $a \in \mathbb{F}\setminus S^\perp$.
    We express the indicator function of the subspace $G^\perp$ using its Fourier expansion over $\mathbb{F}$ 
    
    \[
        \mathbf{1}_{G^\perp}(y) = \frac{1}{|G|} \sum_{g \in G} \chi_g(y).
    \]   
    Then,
    \begin{equation}
        \label{eq:lambda_bound}
        \lambda_a = \frac{1}{|H|} \sum_{h \in H} \mathbf{1}_{G^\perp}(h^{-1}a) = \frac{1}{|H|} \sum_{h \in H} \frac{1}{|G|} \sum_{g \in G} \chi_g(h^{-1}a) = \frac{1}{|G|} \sum_{g \in G} \left( \frac{1}{|H|} \sum_{h \in H} \chi_g(h^{-1}a) \right).
    \end{equation}
    Note that $G\cap (H\cdot a)^\perp$ is a proper $\mathbb{F}_p$ subspace of $G$. Otherwise $G\subseteq (H\cdot a )^\perp$, which would imply that $a\in S^\perp$. 
    For each $g\in G\cap (H\cdot a)^\perp$ the inner sum is $1$, and there are at most $\abs{G}/p$ such $g$'s, therefore, by the triangle inequality, the RHS of \eqref{eq:lambda_bound} is at most 
    $$\frac{1}{p}+\frac{1}{\abs{G}}\sum_{g\in G\setminus (H\cdot a)^\perp} \frac{1}{|H|} \abs{\sum_{h \in H} \chi_g(h^{-1}a)}= 
    \frac{1}{p}+\frac{1}{\abs{G}}\sum_{ga\in G\cdot a\setminus H^\perp} \frac{1}{|H|} \abs{\sum_{h \in H} \chi_{ga}(h^{-1})}
    \leq \frac{1}{p}+\frac{M}{|H|}.$$
    Since $\sigma_2(T) = \sqrt{\max_{a\in \mathbb{F}\setminus S^\perp  }\lambda_a}$, the result follows.
\end{proof}

\section{First Code Instantiation}
\label{sec:main-construction}
In this section, we present our first construction of algebraic expander codes as a concrete instantiation of the general framework established in Section~\ref{sec:general}. We show that for any desired target local rate, this instantiation yields an infinite family of asymptotically good codes with strong spectral expansion, whose local constraints are given by relatively short Reed--Solomon codes. 

However, the alphabet size grows linearly with the block length, and the maximum degree of the underlying graph grows polynomially in the number of vertices.

\subsection{Code Instantiation}
\label{sec:first-code}
Let $m \ge 2$ be a fixed integer, and let $p$ be a prime. We consider the asymptotic regime where $m$ is fixed and $p \to \infty$. 

We apply the general construction $\mathcal{C}(G, H; r, D)$ given in Section \ref{sec:code-def} with the following parameters.

\begin{itemize}
    \item \textbf{Additive Subgroup ($G$):} Let $g(X) = X^{p^m} + X^p + X$. We define $G$ to be the set of roots of $g(X)$ in its splitting field. Since $g(X)$ is a linearized polynomial with $\gcd(g, g')=1$, $G$ is an abelian group of order $|G| = p^m$.
    We identify $G$ with the group of translations $\{x\mapsto x+a:a\in G\}$, and then by \eqref{eq:additive-subgroup-poly} its associated $G$-invariant polynomial is $g(X)$.
    \item \textbf{Multiplicative Subgroup ($H$):} Let $H = \mathbb{F}_{p^m}^\times$ be the  multiplicative group of order $p^m-1$. 
   We identify $H$ with the  group of scaling maps 
   $\{x\mapsto hx: h\in H\}$, and then by \eqref{eq:mult-subgroup-poly} its associated $H$-invariant polynomial is $h(X)=X^{\abs{H}}$.
\end{itemize}

\paragraph{Structure and Code Length.}
Let $A = \langle G, H \rangle$ be a subgroup of 
$\mathrm{AGL}(1,\mathbb{F})$ and recall that we can assume that  $\mathbb{F}$ is a large enough field extension of $\mathbb{F}_p$ such that  $\abs{\mathbb{F}}\geq \abs{A}.$

By Theorem \ref{thm:structural-theorem}, the group $A$ is isomorphic to $S \rtimes H$, where $S = \sum_{h\in H}hG$. The following lemma  shows that $\abs{S}=p^{m^2}$, and therefore the code length is
\[
    n = |A| = |S| \cdot |H| = p^{m^2}(p^m - 1)= p^{m^2+m}(1-o(1)),
\]
where $o(1)$ tends to zero as $p\rightarrow \infty.$

\begin{lemma}
\label{lem:moore}
Let $G$ and $H$, as defined above, be additive and multiplicative subgroups of the field $\mathbb{F}$, respectively. Then the set $S=\sum_{h\in H}hG$ has size $p^{m^2}$.
\end{lemma}
\begin{proof}
    See Appendix \ref{Proof of Lemma moore}.
\end{proof}

\subsection{Code Parameters}
In this section we analyze the rate and relative distance of the code. 
Recall that $m \ge 2$ is fixed. To simplify the analysis, we work in the asymptotic regime where     $p \to \infty$.

Fix $D\leq n$ and $r\in (0,1)$, and let 
$\mathcal{C}=\mathcal{C}(G, H;r,D)$ as in \eqref{eq:main-code-construction}.
 Define the normalized degree bound $\rho = \lim_{p\to\infty} D/n$, where clearly $\rho \in [0, 1]$.  

We first provide the parameters of the code in the important case of $D=n$, i.e., $\rho=1$. 

\begin{theorem}
\label{thm:optimized_construction2}
Let $m \ge 2$ and $r \in (0, 1)$ be fixed. The code $\mathcal{C}$ is an explicit family of expander codes of length $n$ which, as $p \to \infty$, has:
\begin{enumerate}
    \item Left and right local code lengths $|G| = p^m$ and $|H| = p^m-1$ (both scaling as $\Theta(n^{\frac{1}{m+1}})$), with local rate at most $r+o(1)$ and local relative distance at least $1-r-o(1)$.
    \item Global relative distance $\delta \ge (1-r)^2$.
    \item Global rate $R \ge \frac{r^{2m+1}}{(2m+1)!}$.
\end{enumerate}
\end{theorem}

The proof of Theorem~\ref{thm:optimized_construction2} follows from the next theorem by setting $\rho=1.$

Next, we give the parameters of the code in the general case of $D\leq n$. 

\begin{theorem}
\label{thm:optimized_construction}
Let $m \ge 2$ be fixed. Let $r \in (0, 1)$ and $\rho \in (0, 1]$. The code $\mathcal{C}$ is an explicit family of expander codes of length $n$ which, as $p \to \infty$, has:
\begin{enumerate}
    \item Left and right local code lengths $|G| = p^m$ and $|H| = p^m-1$ (both scaling as $\Theta(n^{\frac{1}{m+1}})$), with local rate at most $r+o(1)$ and local relative distance at least $1-r-o(1)$.
    \item Global relative distance $\delta \ge \max(1-\rho, (1-r)^2)$.
    \item Global rate 
    \[ R \ge \frac{1}{(C+1)!} \left(r^{C+1} - \max(0, r-\rho)^{C+1}\right), \] 
    where $C=2m$.
\end{enumerate}
\end{theorem}
\begin{proof}
The claims about the local rate and relative distance follow from the code construction and the fact that the local codes are RS codes. We proceed to establish the bounds on the global rate and relative distance. 

\paragraph{Distance Bound.}
We utilize both the algebraic definition of the code and the combinatorial properties of its underlying graph.

\textit{Algebraic Bound:} The global degree constraint ensures that $\deg(f) < D\leq n$ for all $f \in \mathcal{W}(r, D)$. Then, since the evaluation map is injective, the relative distance is bounded by $\delta \ge 1 - D/n$. Asymptotically, this implies that $\delta \ge 1-\rho$.

\textit{Combinatorial Bound:} View the code as an expander code with local relative distance $\delta_0 = 1-r$. In Theorem~\ref{thm:spectral-bound} (below) we show that the coset graph $\mathsf{B}(G,H)$ is a strong spectral expander; in particular, the normalized second largest singular value of its bi-adjacency matrix satisfies $\lambda(\mathsf{B}) = O(1/\sqrt{p})$. Hence, the global relative distance of the code satisfies
$$ \delta \ge \delta_0 (\delta_0 - \lambda(\mathsf{B})). $$
As $p \to \infty$, $\lambda(\mathsf{B}) \to 0$, yielding the asymptotic bound $\delta \ge (1-r)^2$.

Combining these yields $\delta \ge \max(1-\rho, (1-r)^2)$.

\paragraph{Rate Bound.}
The rate analysis is slightly more technical, as we need to lower bound  the dimension of the subspace $\mathcal{W}_{r,D}$. Towards this end we show that many polynomials of the form $g(X)^iX^j$ have small $h$-base degree. 

We begin first with the analysis of the degree of $g(X)$ in $h$-base . Note that $g(X) = X h(X) + X^p + X$, and therefore $\deg_h(g)=p.$ Next, we  establish the $h$-base degree of $g$ under  the Frobenius automorphism.

\begin{lemma}
\label{lem:h_degree_of_g_frobenius}
Let $k \ge 0$ be an integer. Then,
\[
W_k:=\deg_h(g^{p^k}) =
\begin{cases}
    p^{(k \mod m) + 1} &  k \mod m \neq m-1, \\
    p^{m-1} & k \mod m = m-1.
\end{cases}
\]
\end{lemma}

\begin{proof}
Write $k = am + b$, where $0 \leq b < m$. Let $Q = \frac{p^{am}-1}{p^m-1}$. We use the identity $X^{p^{am}} = X \cdot h(X)^Q$. Thus, $X^{p^k} = (X^{p^{am}})^{p^b} = X^{p^b} h(X)^{Q p^b}$.

We compute $g(X)^{p^k} = (X h(X) + X^p + X)^{p^k} = X^{p^k} h(X)^{p^k} + X^{p^{k+1}} + X^{p^k}$.
Substituting the expressions for $X^{p^k}$:
\begin{align*}
g(X)^{p^k} &= \left(X^{p^b} h(X)^{Q p^b}\right) h(X)^{p^k} + \left(X^{p^{b+1}} h(X)^{Q p^{b+1}}\right) + \left(X^{p^b} h(X)^{Q p^b}\right) \\
&= X^{p^b} \left( h(X)^{Q p^b + p^k} + h(X)^{Q p^b} \right) + X^{p^{b+1}} h(X)^{Q p^{b+1}}.
\end{align*}

We analyze the degrees of the coefficients in the $h$-base expansion.

\textbf{Case 1:} $b < m-1$. The degrees of the coefficients are $p^b$ and $p^{b+1}$. Since $p^{b+1} < p^m = \deg(h)+1$, these coefficients are already reduced. The maximum degree is $p^{b+1} = p^{(k \mod m) + 1}$.

\textbf{Case 2:} $b = m-1$. The degrees are $p^{m-1}$ and $p^m$. We must reduce the term $X^{p^m}$ using $X^{p^m} = X h(X) $ we get 
\begin{align*}
g(X)^{p^k} &= X^{p^{m-1}} \left( h(X)^{Q p^{m-1} + p^k} + h(X)^{Q p^{m-1}} \right) + (X h(X) ) h(X)^{Q p^m}.
\end{align*}
The degrees of the coefficients are $p^{m-1}$ and $1$. The maximum degree is $p^{m-1}$.
\end{proof}

Note that the weights $W_k $ defined in Lemma \ref{lem:h_degree_of_g_frobenius} for $k\geq 0$ form a periodic sequence of period length length $m$: $p^1, \dots, p^{m-2}, p^{m-1}, p^{m-1}$. Also, the dominant weight is $p^{m-1}$. 

Next, note that any polynomial of the form $g(X)^iX^j$ such  that
\begin{enumerate}
    \item $ 0\leq j < r\cdot \deg (g)=rp^m$
    \item $ \deg_h\!\big(g(X)^i X^j\big) < r\cdot \deg (h)=r(p^m-1)$ 
        \item $\deg\!\big(g(X)^i X^j\big)=ip^m+j < D$
\end{enumerate}
belongs to the subspace $\mathcal{W}_{r,D}$ \eqref{eq:subspace-def}. Next, we would like to bound the number of such polynomials. Note that all such polynomials are of distinct degree, and therefore linearly independent, which implies a lower bound on the dimension of $\mathcal{W}_{r,D}.$

Next, we define a simplified set of constraints 
\begin{enumerate}
      \item \textbf{Local Constraint:} $\sum_ki_kW_k + j < r (p^m-1)$, where $i=\sum_ki_kp^k$ is the $p$-adic expansion of $i.$
    \item \textbf{Global Constraint:} $\deg(g^i X^j) = i \cdot p^m + j=\sum_k i_kp^{k+m}+j < D$.
 
\end{enumerate}
Note that any polynomial $g(X)^iX^j$ that satisfies the local and global constraints also satisfies the three constraints $(1)-(3)$, above, and therefore they belong to $\mathcal{W}(r,D).$ Indeed, the global constraint is identical to the last constraint (3). Furthermore, if it satisfies the local constraint, then necessarily $j< rp^m$. Lastly, 
$$\deg_h(g(X)^iX^j)=
\deg_h\left(\prod_k \left(g(X)^{p^k}\right)^{i_k} X^j\right)
\leq \sum_ki_kW_k+j<r(p^m-1),$$
where the first inequality follows from the subadditivity of the $h$-base degree (Lemma \ref{lem:base-u-subadditive}) and Lemma \ref{lem:h_degree_of_g_frobenius}. 

Next, we lower bound the number of monomials that satisfy the local and global  constraints. 
The global degree constraint implies $i < D/p^m \approx \rho p^{m^2}$, restricting the $p$-adic expansion of $i$ to indices $0 \le k < m^2$.

 We define the normalized variables $x_k := i_k/p \in [0,1)$ and $z := j/p^m \ge 0$, and analyze the constraints asymptotically as $p\to\infty$.

\begin{enumerate}
\item \textbf{Asymptotic local constraint:}
Let $I_{\mathrm{dom}}$ be the set of indices corresponding to the dominant weight $p^{m-1}$, namely
\[
I_{\mathrm{dom}}
:=
\{\,0\le k<m^2 : W_k=p^{m-1}\,\},
\qquad\text{and } |I_{\mathrm{dom}}|=2m.
\]
Let $C:=2m$. Dividing the local constraint by $p^m$ and taking the limit $p\to\infty$ yields
\begin{equation}
\label{eq:asy-local-const}
\sum_{k\in I_{\mathrm{dom}}} x_k + z < r.
\end{equation}

\item \textbf{Asymptotic global constraint:}
Let $k_{\max}=m^2-1$. Dividing the global constraint by $p^{m^2+m}$ and taking the limit $p\to\infty$ gives
\begin{equation}
\label{eq:asy-global-const}
x_{k_{\max}} < \rho.
\end{equation}
\end{enumerate}

Note that $k_{\max}\in I_{\mathrm{dom}}$, since by Lemma~\ref{lem:h_degree_of_g_frobenius} we have
$W_{k_{\max}}=W_{m^2-1}=p^{m-1}$.
Thus, the asymptotic global constraint~\eqref{eq:asy-global-const} restricts a variable that also appears in the asymptotic local constraint~\eqref{eq:asy-local-const}.

To lower bound the dimension of $\mathcal{W}_{r,D}$, we count the number of valid discrete tuples $(j, \{i_k\})$. As $p \to \infty$, this discrete counting problem can be approximated by a continuous volume via a standard multidimensional Riemann-sum argument.
The $m^2-C$ subdominant variables $x_k$ (for $k \notin I_{\mathrm{dom}}$) vanish from the asymptotic constraints \eqref{eq:asy-local-const} and \eqref{eq:asy-global-const}. Thus, they are asymptotically unconstrained and contribute a factor of $1$ to the normalized volume.
The fraction of valid tuples therefore converges exactly to the continuous volume of the polytope $\mathcal{P}$ defined by the dominant variables $\{x_k\}_{k\in I_{\mathrm{dom}}}$ and $z$.
Since the total number of unconstrained tuples is $p^m \cdot p^{m^2} = p^{m^2+m}$, the number of admissible monomials is asymptotically $p^{m^2+m} \mathrm{Vol}(\mathcal{P})$.
 
The next lemma computes the volume of $\mathcal{P}$.

\begin{lemma}
\label{lem:polytope_volume}
The volume of the polytope $\mathcal{P}$ defined by constraints \eqref{eq:asy-local-const} and \eqref{eq:asy-global-const} is
$$ \text{Vol}(\mathcal{P}) = \frac{1}{(C+1)!} \left(r^{C+1} - \max(0, r-\rho)^{C+1}\right). $$
\end{lemma}

\begin{proof}
    See Appendix \ref{proof-lemma-polytope-volume}.
\end{proof}
Since the code length is $n = p^{m^2+m}(1-o(1))$, dividing the asymptotically number of admissible monomials  $p^{m^2+m}\mathrm{Vol}(\mathcal{P})$ by $n$ shows that the global rate $R$ as $p \to \infty$ is bounded below by $\mathrm{Vol}(\mathcal{P})$. Hence, by Lemma \ref{lem:polytope_volume}, we obtain that
\[
R \ge \frac{1}{(C+1)!} \left(r^{C+1} - \max(0, r-\rho)^{C+1}\right). 
\]
\end{proof}

From Theorem \ref{thm:optimized_construction} we identify three regions where the parameters of the code behave differently, depending on the relation between $r$ and $\rho$. We summarize it in the following corollary.

\begin{corollary}
\label{cor:optimal_tradeoff}
For a fixed $m \ge 2$, a  local rate $r \in (0, 1)$ and $\rho\in (0,1)$, the code $\mathcal{C}$ satisfies 
$$
\begin{aligned}
R &\ge
\begin{cases}
   \frac{r^{C+1} - (r-\rho)^{C+1}}{(C+1)!} & \rho\le r,\\
   \frac{r^{C+1}}{(C+1)!} & r<\rho.\\
\end{cases}\\
\delta &\ge
\begin{cases}
    1-\rho & \rho\le 2r-r^2,\\
     (1-r)^2 & 2r-r^2<\rho.
\end{cases}
\end{aligned}
$$
\end{corollary}
By inspecting Corollary~\ref{cor:optimal_tradeoff}, we observe an unsatisfying behavior of our rate bound. Intuitively, one expects that increasing $\rho$ should increase the global rate $R$ of the code. However, once $\rho$ exceeds $r$, our lower bound on the rate no longer improves. This raises the question of whether the bound is simply not tight, and in fact the true rate is increasing in $\rho$ throughout the admissible range, or whether the global rate indeed does not increase beyond this point.
\subsection{Spectral Expansion}
In this section we establish the final ingredient needed to conclude that our code is an expander code, namely that the coset graph associated with the instantiation of Section~\ref{sec:first-code} is expanding. In particular, we prove that this graph is a strong spectral expander, with second singular value on the order of $O(1/\sqrt{p})$.

\begin{theorem}
\label{thm:spectral-bound}
The second largest singular value of the normalized adjacency operator of the coset graph associated with the code instantiation in Section \ref{sec:first-code}  satisfies,
    \[
        \sigma_2(T) \le \sqrt{\frac{1}{p}+\frac{1}{p^m-1}}.
    \]
\end{theorem}

\begin{proof}
 The result will follow by 
 Proposition \ref{prop:general-bound}, and by showing that  
 \[
        \max_{a\in \mathbb{F}\setminus H^\perp } \left| \sum_{h \in H} \chi_a(h) \right|=1.
    \] 
 Indeed, for any $a\in \mathbb{F}\setminus H^\perp$, orthogonality of additive characters over $\mathbb{F}_{p^m}$ gives
$$
\sum_{h\in \mathbb{F}_{p^m}}\chi_a(h)=0,
$$
so
$$
\abs{\sum_{h\in H}\chi_{a}(h)}
=\abs{\sum_{h\in \mathbb{F}_{p^m}}\chi_{a}(h)-\chi_a(0)}
=1.
$$
 \end{proof}
\section{Second Code Instantiation}
\label{sec:second-instantiation}
In this section, we give another instantiation of the code construction from Section~\ref{sec:general}. 
Recall that in the first instantiation (Section~\ref{sec:first-code}) the two degrees in the graph are as close as possible (they differ by one). Moreover, it exploits the choice $H=\mathbb{F}_{p^m}^\times$, which yields an especially clean and strong character-sum bound and hence excellent spectral expansion. 
However, this choice is relatively rigid, as it essentially fixes the right degree to be of the form $|H|=p^m-1$, since $H$ must be the full multiplicative group of a field.

Furthermore, the field over which the code is defined must contain the splitting field of the linearized polynomial $g(X)$ that defines $G$, which can be large.

In this section we present a second instantiation that addresses these two caveats, thereby exemplifying the flexibility of the general construction. 
In particular, in this instantiation we take $G=\mathbb{F}_{p^m}$ (hence algebraically simpler than in the first instantiation) and choose $H$ to be a large cyclic subgroup of $\mathbb{F}_{p^{m+1}}^\times$ of constant index $1/\gamma$.
This allows us to tune the right degree $|H|=\gamma(p^{m+1}-1)$ while still obtaining strong spectral expansion via standard Gauss-sum estimates for multiplicative subgroups.
The result is an explicit expander-code family with the same qualitative distance guarantees, together with a parameter $\gamma$ that controls the locality/degree tradeoff on the $H$-side.

We proceed with the code instantiation.
\subsection{Code Instantiation}
\label{sec:second-code-instan}
Fix an integer $m\ge 2$ and let $p$ be a prime that we later let tend to infinity.
Let $q:=p^{m+1}$.
We apply the general construction $\mathcal{C}(G, H; r, D)$ given in Section \ref{sec:code-def} with the following parameters.

\begin{itemize}
    \item \textbf{Additive Subgroup:} Let $G:=\mathbb F_{p^m}$.

    \item \textbf{Multiplicative Subgroup:}
Let $\gamma\in (0,1)$ be a reciprocal of a fixed positive integer such that $\gamma(p-1)$ is an integer\footnote{An infinite sequence of such primes $p$  exists. Indeed,  if $\gamma=1/a$ for some fixed positive integer $a$, then by 
Dirichlet’s theorem on primes in arithmetic progressions there are  infinitely many primes 
$p\equiv 1 \pmod{a}$. For such primes, clearly $\gamma(p-1)$ is an integer.
}, and let $H\le \mathbb F_q^\times$ be the unique cyclic subgroup of order
\[
|H| = \gamma\cdot (p^{m+1}-1).
\]
Note that since $\gamma(p-1)$ is an integer, the size of $|H|$ is also an integer.
\end{itemize}
 As before, we identify $G$ with the translation group
$\{x\mapsto x+g : g\in G\}$ and $H$ with the scaling group
$\{x\mapsto hx : h\in H\}$. The associated invariant polynomials are therefore
\[
g(X)=\prod_{a\in \mathbb F_{p^m}} (X-a)=X^{p^m}-X, \text{ and }
h(X)=X^{|H|}.
\]

\paragraph{Structure and Code Length.}
Let $\mathbb{F}$ be a large enough field extension of $\mathbb{F}_p$ such that the group 
$A = \langle G, H \rangle$ is a subgroup of 
$\mathrm{AGL}(1,\mathbb{F})$, and $\abs{\mathbb{F}}\geq \abs{A}.$
Let 
$
S:=\sum_{h\in H} hG$ be the additive 
closure of $G$ under the action of $H$. 
The following lemma characterizes $S.$

\begin{lemma}
\label{lem:S-second}
The set $S$ satisfies \[
S = \mathbb F_{p^{m(m+1)}}.
\]
\end{lemma}

\begin{proof}
Note that the smallest power of $p$ that is at least $\abs{H}$ is $p^{m+1}$. Therefore,  since $H\subset \mathbb{F}_{p^{m+1}}$ it holds that $\mathbb{F}_p(H)=\mathbb{F}_{p^{m+1}}.$
Then, by Lemma \ref{lem:S-span} 
\[
S=\Span_{\mathbb{F}_p(H)}(G)
=\mathrm{Span}_{\mathbb F_{p^{m+1}}}(G)=
\mathrm{Span}_{\mathbb F_{p^{m+1}}}(\mathbb{F}_{p^m})=\mathbb{F}_{p^{m(m+1)}},\]
where the last equality follows since $\gcd(m,m+1)=1$, $\mathbb{F}_{p^m}\cap \mathbb{F}_{p^{m+1}}=\mathbb{F}_p$.
Hence any $\mathbb{F}_p$-basis of $\mathbb{F}_{p^m}$ is $\mathbb{F}_{p^{m+1}}$-linearly independent, so
$\dim_{\mathbb{F}_{p^{m+1}}}\left(\Span_{\mathbb{F}_{p^{m+1}}}(\mathbb{F}_{p^m})\right)=m$.
But $[\mathbb{F}_{p^{m(m+1)}}:\mathbb{F}_{p^{m+1}}]=m$, so this span is an
$m$-dimensional $\mathbb{F}_{p^{m+1}}$-subspace of $\mathbb{F}_{p^{m(m+1)}}$, and therefore equals the whole field.
\end{proof}
Recall that the code length $n=|A|=|S||H|$, where 
$
A=\langle G,H\rangle \cong S\rtimes H.
$
Hence, by Lemma \ref{lem:S-second}
\[
n=|S||H| = p^{m(m+1)}\cdot \gamma\cdot (p^{m+1}-1)= \gamma\cdot p^{(m+1)^2}\cdot (1-o(1)),
\]
where $o(1)$ tends to zero as $p\to\infty$.
\begin{remark}
For concreteness, we specify the  field over which the code is defined. Let
$
\mathbb{F} = \mathbb{F}_{p^{2m(m+1)}},
$
so that $|\mathbb{F}| \ge |A|$ and since $S,H\subseteq \mathbb{F}$, $A$ is as a subgroup of $\mathrm{AGL}(1,\mathbb{F})$. 
\end{remark}
Choose an element $\alpha \in \mathbb{F}$ that is not stabilized by any nontrivial affine transformation in $A$,  and let
\[
\Omega := \{ a(\alpha) : a \in A \} \subseteq \mathbb{F},
\]
be the orbit of $\alpha$ under the  action of $A$. We take this orbit as the evaluation set of the code.

We conclude that  the resulting code has length
$
n  = \Theta\!\left(p^{(m+1)^2}\right)
$, 
over a field of size $|\mathbb{F}| = p^{2m(m+1)}$.

\subsection{Code Parameters}
Fix $r\in(0,1)$ and $D\le n$, and let $\rho=\lim_{p\to\infty} D/n\in(0,1)$.
Let $\mathcal{C}=\mathcal{C}(G,H;r,D)$ be the code from~\eqref{eq:main-code-construction}.

We first provide the parameters of the code in the important case of $D=n$, i.e., $\rho=1$.

\begin{theorem}
\label{thm:params-second2}
The code $\mathcal{C}$ is an explicit family of expander codes of length $n$ which, as $p \to \infty$, has:
\begin{enumerate}
    \item Left local code length $|G|=p^m = \Theta(n^{\frac{m}{(m+1)^2}})$ and right local code length $|H|=\gamma (p^{m+1}-1) = \Theta(n^{\frac{1}{m+1}})$, with local rate at most $r+o(1)$ and local relative distance at least $1-r-o(1)$.
    \item Global relative distance $\delta \ge (1-r)^2$.
    \item Global rate bounded below by \[ R\ge \frac{\gamma^{2m}\, r^{2m+2}}{(2m+1)!}. \]
\end{enumerate}
\end{theorem}

The proof of Theorem~\ref{thm:params-second2} follows from Theorem~\ref{thm:params-second} below by setting $\rho=1$.

The following theorem summarizes the parameters of the code in the general case of $D\leq n$. 

\begin{theorem}
\label{thm:params-second}
The code $\mathcal{C}$ is an explicit family of expander codes of length $n$ which, as $p \to \infty$, has:
\begin{enumerate}
    \item Left local code length $|G|=p^m = \Theta(n^{\frac{m}{(m+1)^2}})$ and right local code length $|H|=\gamma (p^{m+1}-1) = \Theta(n^{\frac{1}{m+1}})$, with local rate at most $r+o(1)$ and local relative distance at least $1-r-o(1)$.
    \item Global relative distance $\delta \ge \max(1-\rho, (1-r)^2)$.
    \item Global rate bounded below by \[ R \ge \frac{\gamma^{2m}\, r}{(2m+1)!}\left(r^{2m+1}-\max(0,r-\rho)^{2m+1}\right). \]
\end{enumerate}
\end{theorem}
\begin{proof}
Items (1) and the RS local distance bound are immediate from the locality analysis in
Section~\ref{sec:local-codes}.

\paragraph{Distance Bound.}
For (2), we combine the algebraic bound $\delta\ge 1-\rho$ (since $\deg(f)<D$ for $f\in\mathcal W_{r,D}$ and evaluation is injective) with the standard expander-code distance bound applied to the underlying coset graph.
By Theorem~\ref{thm:spectral-bound-second} below, the graph satisfies $\sigma_2(T)=o(1)$ as $p\to\infty$, and hence asymptotically $\delta\ge (1-r)^2$.
Combining the two bounds gives $\delta \ge \max(1-\rho,(1-r)^2)$.
\paragraph{Rate Bound.}
 For (3), we lower bound $\dim(\mathcal W_{r,D})$ by exhibiting many  polynomials of the form 
$g(X)^iX^j$ that are in $\mathcal W_{r,D}$ and they  all have distinct degrees. 

Consider polynomials $g(X)^iX^j$ with $0\le j<rp^m$ and $ip^m+j<D$.
Then, clearly $\deg_g(g^iX^j)=j<r|G|$.
Write $i=\sum_{k=0}^{k_{\max}} i_k p^k$ with $0\le i_k<p$, where
$
k_{\max} = m(m+1)$ since $i< D/p^m \leq p^{m^2+m+1}$.

For $k\geq 0$ define the weight
\[
W_k:=\deg_h\!\big(g^{p^k}\big),
\] to be the $h$-base degree of the polynomial $g^{p^k}.$ The following lemma gives a precise description of the weights. 

\begin{lemma}
\label{lem:weights-second}
For a nonnegative integer  $k$ 
\[
W_k=
\begin{cases}
p^m & k \bmod (m+1) = 0,\\
p^{k\mod(m+1)} & \text{ else}.
\end{cases}
\]
\end{lemma}

By the subadditivity of the $h$-base degree (Lemma~\ref{lem:base-u-subadditive}),
\[
\deg_h(g^iX^j)\ \le\ j+\sum_k i_k W_k.
\]
Thus it suffices to consider the simplified constraint
\[
j+\sum_k i_k W_k < r|H|.
\]
We conclude that any polynomial of the form $g(X)^iX^j$ such that 
\begin{enumerate}
    \item $j<rp^m$ 
    \item \textbf{Global Constraint:} $ip^m+j<D$
    \item \textbf{Local Constraint:}  $j+\sum_k i_k W_k < r|H|$,
\end{enumerate}
is in $\mathcal{W}_{r,D}.$ Next, we give a lower bound on the number of such polynomials. 

Let $I_{\mathrm{dom}}:=\{0\le k\le k_{\max}: k\bmod(m+1)\in\{0,m\}\}$.
By Lemma~\ref{lem:weights-second}, $W_k=p^m$ for $k\in I_{\mathrm{dom}}$ and $W_k\le p^{m-1}$ otherwise.
Clearly,  $|I_{\mathrm{dom}}|=2m+1$.

We define the normalized variables $x_k := i_k/p \in [0,1)$ for $k\in I_{\mathrm{dom}}$, and \begin{equation}
    \label{eq:asy-local-param}
z := j/p^m \in[0,r),
\end{equation} and analyze the constraints asymptotically as $p\to\infty$.

By dividing the  local constraint by $p^{m+1}$ and letting  $p\to\infty$ we get the asymptotic local constraint (as  subdominant terms vanish)
\begin{equation}
\label{eq:asy-local-const2}    
\sum_{k\in I_{\mathrm{dom}}} x_k \ <\ r\gamma,
\end{equation}
since $ |H|/p^{m+1}\to\gamma$ and $j/p^{m+1}\to 0$.

Similarly, by dividing the global  constraint by $p^{m^2+2m+1}$ and letting $p\to\infty$ we get the  asymptotic global constraint
\begin{equation}
\label{eq:asy-globacl-const2}
x_{k_{\max}} < \rho\gamma.
\end{equation}
As before, we approximate the discrete count of valid tuples via a standard Riemann-sum argument.
The $m^2-m$ subdominant variables $x_k$ vanish from the asymptotic limits and therefore are unconstrained in the interval $[0,1)$, contributing a factor of $1$ to the volume. Thus, the fraction of valid tuples converges to the geometric volume of the polytope $\mathcal{P}$ defined by the constraints \eqref{eq:asy-local-param}, \eqref{eq:asy-local-const2} and \eqref{eq:asy-globacl-const2} on the dominant variables $\{x_k\}_{k\in I_{\mathrm{dom}}}$ and $z$.
Since the total number of unconstrained tuples is $p^m \cdot p^{m^2+m+1} = p^{m^2+2m+1}$, the asymptotic dimension of $\mathcal{W}_{r,D}$ is $p^{m^2+2m+1}\mathrm{Vol}(\mathcal{P})$.

The next lemma computes the  volume of $\mathcal{P}$.
\begin{lemma}
\label{lem:polytope_volume2}
The volume of the polytope $\mathcal{P}$ defined by the constraints \eqref{eq:asy-local-param}, \eqref{eq:asy-local-const2} and \eqref{eq:asy-globacl-const2} is
\[
\mathrm{Vol}(\mathcal{P}) = \frac{\gamma^{2m+1}r\left(r^{2m+1}-\max(0,r-\rho)^{2m+1}\right)}{(2m+1)!}.
\]
\end{lemma}

\begin{proof}
    See Appendix \ref{proof-of-polytope-volume}.
\end{proof}

Since the block length is $n\sim \gamma p^{m^2+2m+1}$, dividing the asymptotic dimension $p^{m^2+2m+1}\mathrm{Vol}(\mathcal{P})$ by $n$ yields the stated lower bound on the global rate $R \ge \frac{1}{\gamma}\mathrm{Vol}(\mathcal{P})$.
\end{proof}

Next, we prove Lemma \ref{lem:weights-second}.
\begin{proof}[Proof of Lemma \ref{lem:weights-second}] Write $k=a(m+1)+t$, where $t=k\mod (m+1).$ 
Since $|H|$ divides $p^{m+1}-1$, we have the identity
\[
X^{p^{m+1}} = X^{p^{m+1}-1}\cdot X = X^{|H|\cdot \frac{p^{m+1}-1}{|H|}}\cdot X
= X\cdot h(X)^{\frac{p^{m+1}-1}{|H|}}
= X\cdot h(X)^{\frac{1}{\gamma}}.
\]
Then by induction on $a$, for any positive integer $a$ we have
\[
X^{p^{a(m+1)}}=X\cdot h(X)^{(\text{some integer})}.
\]
Raising both sides to $p^t$ shows that $X^{p^k}$ can be written in base $h$ as
\[
X^{p^k}=X^{p^t}\cdot h(X)^{(\text{some integer})},
\]
with coefficient degree $\deg(X^{p^t})=p^t\leq p^m<\deg(h)$.

Now
\begin{equation}
\label{eq:weights}
W_k=\deg_h\left(g^{p^k}\right) = \deg_h\left(X^{p^{m+k}}-X^{p^k}\right).
\end{equation}

If $t=0$, then $k=a(m+1)$ and $m+k=a(m+1)+m$, so the RHS of \eqref{eq:weights} becomes
$$\deg_h\left(
X^{p^m}\cdot h(X)^{(\text{some integer})}-X\cdot h(X)^{(\text{some integer})}
\right)=p^m.$$

If $1\le t\le m$, then $m+k=a(m+1)+m+t=(a+1)(m+1)+(t-1)$,
so the RHS of \eqref{eq:weights} becomes
$$\deg_h\left(
X^{p^{(t-1)}}\cdot h(X)^{(\text{some integer})}-X^{p^t}\cdot h(X)^{(\text{some integer})}
\right)=p^t.$$
\end{proof}

\subsection{Spectral Expansion via Gauss Sums}
In this section we show that the coset graph underlying the code instantiation of Section~\ref{sec:second-code-instan} is a strong spectral expander, with second singular value on the order of $O(1/\sqrt{p})$. This expansion bound follows by applying Proposition~\ref{prop:general-bound}. 

To this end, it suffices to bound the additive character sums
\[
\left|\sum_{h\in H}\chi_a(h)\right|
\qquad\text{for } a\notin H^\perp.
\]
We obtain the required estimate via a standard application of Gauss sums, as we now explain.

\begin{lemma}
\label{lem:gauss-subgroup}
Let $H\le \mathbb F_q^\times$ be a multiplicative subgroup.
Then for every nontrivial additive character $\psi$ of $\mathbb F_q$,
\[
\left|\sum_{h\in H}\psi(h)\right|\ \le\ \sqrt{q}.
\]
\end{lemma}

\begin{proof}
Let   $s=(q-1)/|H|$ be the index of $H$, and fix a multiplicative character $\chi$ of $\mathbb F_q^\times$ of exact order $s$.
Then, the indicator of $H$ satisfies
\[
\mathbf{1}_H(x) = \frac1s\sum_{j=0}^{s-1}\chi(x)^j \qquad (x\in\mathbb F_q^\times).
\]
Hence
\[
\sum_{h\in H}\psi(h)
=\sum_{x\in\mathbb F_q^\times}\mathbf{1}_H(x)\psi(x)
=\frac1s\sum_{j=0}^{s-1}\sum_{x\in\mathbb F_q^\times}\chi(x)^j\psi(x).
\]
The $j=0$ term equals $\sum_{x\in\mathbb F_q^\times}\psi(x)=-1$ since $\psi$ is nontrivial.
For $j\ne 0$ the inner sum is a Gauss sum and has magnitude $\sqrt{q}$.
Thus
\[
\left|\sum_{h\in H}\psi(h)\right|
\le \frac1s\Bigl(1+(s-1)\sqrt{q}\Bigr)\le \sqrt{q}.
\]
\end{proof}

\begin{theorem}
\label{thm:spectral-bound-second}
Let $T$ be the normalized bi-adjacency operator of the coset graph $\mathsf{B}(G,H)$.
Then
\[
\sigma_2(T)\ \le\ \sqrt{\frac1p+\frac{1+o(1)}{\gamma\, p^{(m+1)/2}}}\;=\;\frac{1+o(1)}{\sqrt{p}}.
\]
\end{theorem}

\begin{proof}
Let  $\chi_a$ be an additive character of $\mathbb{F}$ with $a\in \mathbb{F}\setminus H^\perp$. Therefore, $\chi_a$  is a nontrivial character since $a\notin H^\perp$. Hence, by Lemma~\ref{lem:gauss-subgroup} we have
\[
\left|\sum_{h\in H}\chi_a(h)\right|\ \le\ \sqrt{q}\Rightarrow M=\max_{a\in \mathbb{F}\setminus H^\perp}\left|\sum_{h\in H}\chi_a(h)\right|\ \le\ \sqrt{q}.
\]
Plugging into Proposition~\ref{prop:general-bound} gives
\[
\sigma_2(T)\le \sqrt{\frac1p+\frac{M}{|H|}} \le \sqrt{\frac1p+\frac{\sqrt{q}}{|H|}}.
\]
Finally use $|H|=\gamma\cdot (q-1)$  to obtain the asymptotic form.
\end{proof}

\section{Concluding Remarks and Open Questions}
\label{sec:conclusion}
In this work, we presented a general framework for constructing algebraic expander codes utilizing the action of subgroups of the affine group $\mathrm{AGL}(1, \mathbb{F})$. By instantiating this framework with non-commuting translation and scaling subgroups, we derived explicit families of codes that circumvent the classical rate barrier for expander codes. Specifically, our constructions achieve a positive global rate even when the local code rate satisfies $r \le 1/2$, a regime where generic expander constructions typically fail.

We conclude by highlighting several open problems and directions for future research arising from this construction.

\begin{enumerate}
    \item \textbf{Tightness of the Rate Bound.}
    The lower bound on the global rate $R(r, m, \rho)$ derived in Theorems \ref{thm:optimized_construction}  and \ref{thm:params-second} establishes the existence of asymptotically good codes, but it exhibits asymptotic behaviors that suggest it is far from tight.
    \begin{itemize}
        \item \emph{Saturation:} The bound saturates once the normalized global degree $\rho$ exceeds the local rate $r$. Intuitively, increasing the allowed global degree $D$ (and thus $\rho$) should strictly increase the dimension of the message space $\mathcal{W}_{r,D}$, yet our polytope volume argument remains constant in this regime.
        \item \emph{Vanishing with Sparsity:} As the sparsity parameter $m$ increases, our lower bound vanishes factorially (scaling roughly as $1/(2m)!$). However, in the regime where the local rate satisfies $r > 1/2$, the standard constraint-counting bound for expander codes (Remark \ref{rem:counting-bound}) guarantees a rate of at least $2r-1$. This standard bound is independent of the graph degree (and thus independent of $m$). The fact that our algebraic bound approaches zero for large $m$, even when $r > 1/2$, indicates that our analysis, which relies on the simultaneous satisfaction of $g$-base and $h$-base degree constraints, likely underestimates the true dimension of the code. Bridging the gap between our bound and the constant $2r-1$ baseline is a key challenge.
    \end{itemize}

    \item \textbf{Optimizing Code Parameters (Degree and Expansion).}
    The coset graphs constructed in this work have a degree that grows polynomially with the block length ($n^{1/m}$). While this improves upon the linear density of commuting-subgroup constructions, achieving \emph{constant} degree (or polylogarithmic degree) remains a major goal for algebraic expander codes.
    \begin{itemize}
        \item Is it possible to instantiate this framework with different subgroups of $\mathrm{AGL}(1, \mathbb{F})$ (or potentially $\mathrm{AGL}(k, \mathbb{F})$) to yield constant-degree graphs?
        \item Furthermore, while we proved that the constructed graphs are strong spectral expanders, determining whether they can achieve the optimal Ramanujan bound is of  interest.
    \end{itemize}

    \item \textbf{Local Testability.}
    The codes constructed here share structural similarities with affine-invariant codes and lifted codes, which are known to support local testing. Given that these codes are defined by local Reed--Solomon constraints on an expander graph, it is natural to ask whether they are also locally testable. 

    \item \textbf{High-Dimensional Expanders (HDX).}
    Our construction utilizes the interaction between two subgroups, $G$ and $H$, to define a bipartite graph (a $1$-dimensional simplicial complex) where edges correspond to group elements and vertices correspond to cosets. A natural generalization is to consider the geometry generated by \emph{three} subgroups $G_1, G_2, G_3 \le \mathrm{AGL}(1, \mathbb{F})$. Can such a construction define a $2$-dimensional simplicial complex (e.g., where faces correspond to elements of $\langle G_1,G_2,G_3\rangle$ and vertices to cosets of the $G_i$'s) that exhibits high-dimensional expansion? If so, placing local codes on the faces of this complex could lead to new algebraic constructions of HDX codes.
\end{enumerate}

\appendix
\section{Proofs of Section \ref{sec:graph-analysis}}
\label{proofs-section-3}

\begin{proof}[Proof of Theorem \ref{thm:structural-theorem}]
We begin by showing that $S$ is an $H$-invariant subspace.
By definition, $S$ is an additive group. Next, we show that $S$ is closed under multiplication by $H$. For any $s = \sum h_k g_k \in S$ and any $h' \in H$, we have $h's = \sum (h'h_k)g_k$. Since $H$ is a group, $h'h_k \in H$, so $h's$ is a sum of elements in $H \cdot G$.

Next, let $\mathcal{A} = \{ \phi_{s,h} : x \mapsto hx + s \mid h \in H, s \in S \}$. We will show that $\mathcal{A}$ is  a group and that $\mathcal{A}=A.$

\paragraph{$\mathcal{A}$ is a Group.}
Consider the composition of $\phi_{s_1, h_1}, \phi_{s_2, h_2} \in \mathcal{A}$:
\[
    (\phi_{s_1, h_1} \circ \phi_{s_2, h_2})(x) = h_1(h_2 x + s_2) + s_1 = (h_1 h_2)x + (h_1 s_2 + s_1).
\]
Since $H$ is a group, $h_1 h_2 \in H$. Since $S$ is an additive group closed under $H$, $h_1 s_2 \in S$, and thus $s' = s_1 + h_1 s_2 \in S$. The inverse of $\phi_{s,h}$ is $\phi_{-h^{-1}s, h^{-1}}$, which is also in $\mathcal{A}$. Thus $\mathcal{A} \le \mathrm{AGL}(1, \mathbb{F})$.

\paragraph{$A = \mathcal{A}$.}
Clearly $G \subset \mathcal{A}$ (via $h=1$) and $H \subset \mathcal{A}$ (via $s=0$). Thus $A \subseteq \mathcal{A}$. Conversely, let $\phi_{s,h}\in \mathcal{A}$. Since $\phi_{s,h}=\phi_{0,h}\circ \phi_{h^{-1}s,1}$, where $\phi_{0,h}\in H$ and $h^{-1}s\in S$, it suffices to show  that for any $ s \in S$, $\phi_{s,1}\in A.$ Next, since an $s\in S$ is a sum of terms $hg$, it suffices to show that the translation $x \mapsto x + hg\in A$. Indeed, it  can be realized as  $\phi_{0,h} \circ \phi_{g,1} \circ \phi_{0, h^{-1}}$, where each transformation is either in $G$ or $H$.  Thus $\mathcal{A} \subseteq A$.

\paragraph{Semidirect Product.}
The mapping $(s, h) \mapsto \phi_{s,h}$ is clearly a bijection between $S \rtimes H$ and $A$. The group operation of $\mathcal{A}$ matches the standard semidirect product definition.
\end{proof}

\begin{proof}[Proof of Lemma \ref{lem:S-span}]
Let
\[
K \;:=\; \Span_{\mathbb F_p}(H)
\;=\;
\left\{\sum_{i=1}^\ell c_i h_i \;:\; \ell\ge 0,\ c_i\in \mathbb F_p,\ h_i\in H\right\}\subseteq \mathbb F.
\]
We first claim that $K$ is a subfield of $\mathbb F$ and hence $K=\mathbb F_p(H)$.

Indeed, $K$ is closed under addition by definition and contains $1$.
It is also closed under multiplication. Indeed,  if $x=\sum_i c_i h_i$ and $y=\sum_j d_j h'_j$, then
\[
xy=\sum_{i,j} (c_i d_j)\,(h_i h'_j),
\]
and since $c_i d_j\in\mathbb F_p$ and $h_i h'_j\in H$,
we have $xy\in K$. Thus $K$ is a finite subring of the field $\mathbb F$.
Being a subring of a field, $K$ has no zero divisors, hence $K$ is a finite integral domain,
and therefore a field. Since $K$ contains $\mathbb F_p$ and $H$, minimality of $\mathbb F_p(H)$ gives
$\mathbb F_p(H)\subseteq K$. For the other direction,  $K\subseteq \mathbb F_p(H)$ because $\mathbb F_p(H)$
is closed under $\mathbb F_p$-linear combinations. Hence $K=\mathbb F_p(H)$.

Next, we prove the desired equality. First, note that every  $hg$ with $h\in H$ and $g\in G$
lies in $\Span_{\mathbb F_p(H)}(G)$ because $h\in \mathbb F_p(H)$. Since
$\Span_{\mathbb F_p(H)}(G)$ is an additive subgroup, it follows that
\[
S=\sum_{h\in H} hG \;\subseteq\; \Span_{\mathbb F_p(H)}(G).
\]

For the reverse inclusion, we show that $S$ is closed under multiplication by scalars from
$\mathbb F_p(H)=K$. By construction, $S$ is an additive subgroup closed under multiplication by  $h\in H$ and multiplication by  $c\in \mathbb F_p$ (since $G$ is an $\mathbb F_p$-subspace).
Therefore, for any $a=\sum_i c_i h_i\in K$ and any $s\in S$,
\[
a s = \left(\sum_i c_i h_i\right)s = \sum_i c_i (h_i s)\ \in\ S,
\]
so $S$ is a $K$-vector space. Since $G\subseteq S$, we conclude that
$
\Span_{K}(G)\ \subseteq\ S.
$
Combining both inclusions gives $S=\Span_{\mathbb F_p(H)}(G)$.
\end{proof}

\section{Proof of Lemma~\ref{lem:moore}}
\label{Proof of Lemma moore}
\begin{proof}
By Lemma~\ref{lem:S-span}, we have $S=\Span_{\mathbb{F}_{p^m}}(G)$.
Next, let $(v_1, \ldots, v_m)$ be a basis of $G$ over $\mathbb{F}_p$. We claim that this basis is also linearly independent over $\mathbb{F}_{p^m}$, which implies
$\abs{\Span_{\mathbb{F}_{p^m}}(G)} = p^{m^2}$.

Suppose, for the sake of contradiction, that \(v_1, \ldots, v_m\) are linearly dependent over \(\mathbb{F}_{p^m}\).  
Let \(t \ge 2\) be the minimal number of vectors in a dependent subset, and assume without loss of generality that \(v_1, \ldots, v_t\) are such a set. Then there exist coefficients \(\alpha_i \in \mathbb{F}_{p^m}\) with
\[
\sum_{i=1}^t \alpha_i v_i = 0, 
\quad \text{and} \quad \alpha_1 = 1.
\]

Applying the polynomial \(g\) (which is \(\mathbb{F}_p\)-linear) to both sides yields
\[
0 
= g\!\left( \sum_{i=1}^t \alpha_i v_i \right) 
= \sum_{i=1}^t \left( \alpha_i^{p^m} v_i^{p^m} + \alpha_i^p v_i^p + \alpha_i v_i \right).
\]
Since each \(v_i\) is a root of \(g\), we have \(v_i^{p^m} + v_i^p + v_i = 0\).  
Subtracting \(\sum_{i=1}^t \alpha_i\big( v_i^{p^m} + v_i^p + v_i \big)\) from the above and using \(\alpha_i^{p^m} = \alpha_i\) (because \(\alpha_i \in \mathbb{F}_{p^m}\)) gives
\[
\sum_{i=1}^t \big( \alpha_i^p - \alpha_i \big) v_i^p = 0.
\]
The coefficient of \(v_1^p\) is zero since \(\alpha_1 = 1\), so we obtain
\begin{equation}
\label{eq:relation-vp}
\sum_{i=2}^t \big( \alpha_i^p - \alpha_i \big) v_i^p = 0.
\end{equation}

Since all the $v_i$ lie in some finite field $\mathbb{F}_{p^s}$, raising \eqref{eq:relation-vp} to the $p^{s-1}$-th power maps $v_i^p$ to $v_i$ and yields
\[
\sum_{i=2}^t \big( \alpha_i^p - \alpha_i \big)^{p^{s-1}} v_i = 0.
\]
By the minimality of \(t\), all coefficients in this relation must be zero, i.e.,
\[
\alpha_i^p - \alpha_i = 0
\quad \text{for all } i \ge 2.
\]
Thus \(\alpha_i \in \mathbb{F}_p\) for all \(i\), including \(\alpha_1 = 1 \in \mathbb{F}_p\).

Therefore, the original dependence relation is an \(\mathbb{F}_p\)-linear dependence among \(v_1, \ldots, v_t\), contradicting the fact that they are part of an \(\mathbb{F}_p\)-basis.  
It follows that \(v_1, \ldots, v_m\) are linearly independent over \(\mathbb{F}_{p^m}\).
\end{proof}

\section{Proof of Lemma \ref{lem:polytope_volume}}
\label{proof-lemma-polytope-volume}
\begin{proof}
We integrate over $y = x_{k_{\max}}$. The range is $0 \le y < \min(r, \rho)$. For a fixed $y$, the volume of the remaining polytope on the remaining $C$ variables  is $\frac{(r-y)^C}{C!}$. We prove it by induction on $C$.  For $C=1$ the region is the interval $[0,r-y)$ of length $r-y$. Assume the claim holds for $C-1$. For $C$ variables, fix $x_1=t\in[0,r-y]$, then  the volume of the remaining polytope on the remaining $C-1$ variables is by the induction hypothesis $\frac{(r-y-t)^{C-1}}{(C-1)!}$. Integrating $t$ gives $$\int_{0}^{r-y}\frac{(r-y-t)^{C-1}}{(C-1)!}\,dt=\frac{(r-y)^C}{C!}.$$ Lastly, the volume of the polytope $\mathcal{P}$ equals 

$$ \text{Vol}(\mathcal{P}) = \int_0^{\min(r, \rho)} \frac{(r-y)^C}{C!} dy = \left[ -\frac{(r-y)^{C+1}}{(C+1)!} \right]_0^{\min(r, \rho)}. $$
Since $r-\min(r, \rho) = \max(0, r-\rho)$, the result follows.
\end{proof}

\section{Proof of Lemma \ref{lem:polytope_volume2}}
\label{proof-of-polytope-volume}

\begin{proof}
Recall that $\mathcal{P}$ is defined over the variables $\{x_k\}_{k\in I_{\mathrm{dom}}}$ and $z$ by
\[
0 \le z < r,\qquad x_k \ge 0\ (k\in I_{\mathrm{dom}}),\qquad
\sum_{k\in I_{\mathrm{dom}}} x_k < r\gamma,\qquad
x_{k_{\max}} < \rho\gamma,
\]
where $|I_{\mathrm{dom}}|=2m+1$ and $k_{\max}\in I_{\mathrm{dom}}$.

First note that $z$ is completely decoupled from the other constraints, and contributes a multiplicative factor of
$
\mathrm{Vol}(\{z: 0\le z<r\}) = r.
$
Hence,
$
\mathrm{Vol}(\mathcal{P}) = r \cdot \mathrm{Vol}(\mathcal{Q})
$
where $\mathcal{Q}$ is the polytope in the $2m+1$ variables $\{x_k\}_{k\in I_{\mathrm{dom}}}$.
For a fixed  $y := x_{k_{\max}}\in (0, \min(r\gamma,\rho\gamma)]$, the volume of the polytope defined by the remaining $2m$ variables 
equals $\frac{(r\gamma-y)^{2m}}{(2m)!}$. The proof follows by induction as in the proof of Lemma \ref{lem:polytope_volume}.

Therefore,
\begin{align*}
\mathrm{Vol}(\mathcal{Q})
&=
\int_{0}^{\min(r\gamma,\rho\gamma)} \frac{(r\gamma-y)^{2m}}{(2m)!}\,dy
=\left[ -\frac{(r\gamma-y)^{2m+1}}{(2m+1)!} \right]_{0}^{\min(r\gamma,\rho\gamma)}\\
&
=
\frac{(r\gamma)^{2m+1} - (r\gamma-\min(r\gamma,\rho\gamma))^{2m+1}}{(2m+1)!}\\
&=\frac{\gamma^{2m+1}\left(r^{2m+1}-\max(0,r-\rho)^{2m+1}\right)}{(2m+1)!},
\end{align*}
where the last equality follows since 
\[
r\gamma-\min(r\gamma,\rho\gamma)=\max(0,r\gamma-\rho\gamma)=\gamma\max(0,r-\rho).
\]
Multiplying by the factor $r$ coming from $z$ yields
the result.
\end{proof}

\bibliographystyle{abbrv}
 \bibliography{reference}
\end{document}